\documentclass[journal,draftcls,onecolumn,12pt]{IEEEtran}%
\usepackage{amsfonts}
\usepackage{amsmath}
\usepackage{amssymb}
\usepackage{graphicx}
\usepackage{cite}%
\providecommand{\U}[1]{\protect\rule{.1in}{.1in}}
\newtheorem{theorem}{Theorem}

\newtheorem{definition}{Definition}

\newtheorem{lemma}{Lemma}

\newtheorem{remark}{Remark}

\begin{document}

\title{Discriminatory\ Lossy Source Coding:\ Side Information Privacy}
\author{Ravi Tandon,~\IEEEmembership{Member,~IEEE,} Lalitha
Sankar,~\IEEEmembership{Member,~IEEE,} and~H. Vincent
Poor,~\IEEEmembership{Fellow,~IEEE}\thanks{R. Tandon, L. Sankar, and H. V.
Poor are with the Department of Electrical\ Engineering at Princeton
University, NJ 08544, USA. email: \{rtandon,lalitha,poor@princeton.edu\}.}%
\thanks{This research was supported in part by the National\ Science
Foundation under Grants CNS-09-05086, CNS-09-05398, and CCF-10-16671, the Air
Force Office of Scientific\ Research under Grant FA9550-09-1-0643, and by a
fellowship from the Princeton University Council on\ Science and Technology. }}
\pubid{~}
\specialpapernotice{~}
\maketitle

\begin{abstract}
A lossy source coding problem is studied in which a source encoder
communicates with two decoders, one with and one without correlated side
information with an additional constraint on the privacy of the side
information at the uninformed decoder. Two cases of this problem arise
depending on the availability of the side information at the encoder. The set
of all feasible rate-distortion-equivocation tuples are characterized for both
cases. The difference between the informed and uninformed cases and the
advantages of encoder side information for enhancing privacy are highlighted
for a binary symmetric source with erasure side information and Hamming distortion.

\end{abstract}


\section{Introduction}

Information sources often need to be made accessible to multiple legitimate
users simultaneously, some of whom can have correlated side information
obtained from other sources or from prior interactions. A natural question
that arises in this context is the following: can the source publish (encode)
its data in a discriminatory manner such that the uninformed user does not
infer the side information, i.e., it is kept private, while providing utility
(fidelity) to both users? Two possible cases arise in this context depending
on whether the encoder is \textit{informed} or \textit{uninformed}, i.e., it
has or does not have access to the correlated side information, respectively.

This question is addressed from strictly a fidelity viewpoint by C. Heegard
and T. Berger in \cite{HeegardBerger}, henceforth referred to as the
Heegard-Berger problem, for the uninformed case and by A. Kaspi \cite{KaspiSI}%
, henceforth referred to as the Kaspi problem, for the informed case wherein
they determined the rate-distortion function for a discrete and memoryless
source pair. Using equivocation as the privacy metric, we address the question
posed above using the source network models in \cite{HeegardBerger} and
\cite{KaspiSI} with an additional constraint on the side information privacy
at the decoder without access to it, i.e., decoder 1 (see Fig. \ref{FigHB}).

We prove here that the encoding scheme for the Heegard-Berger problem achieves
the minimal rate while guaranteeing the maximal equivocation for any feasible
distortion pair at the two decoders when the encoder is uninformed. Informally
speaking, the Heegard-Berger coding scheme involves a combination of a
rate-distortion code and a conditional Wyner-Ziv code which is revealed to
both decoders. Our proof exploits the fact that conditioned on what is
decodable by decoder 1, i.e., the rate-distortion code, the additional
information intended for decoder 2, i.e. the conditional Wyner-Ziv bin index,
is asymptotically independent of the side information, $Y$ (see Fig.
\ref{FigHB}). Observing that the generation of the conditional Wyner-Ziv bin
index is analogous to the Slepian-Wolf binning scheme, we prove this
independence property for both the Slepian-Wolf and the Wyner-Ziv encoding.
Next, we prove a similar independence property for the Heegard-Berger coding
scheme, which in turn allows us to demonstrate the optimality of this scheme
for the problem studied in this paper.

On the other hand, for the informed encoder case, we present a modified coding
scheme (vis-\`{a}-vis the Kaspi scheme) which achieves the set of all feasible
rate-equivocation pairs for the desired fidelity requirements at the two
decoders. The Kaspi coding scheme exploits the encoder side information $Y$
(see Fig. \ref{FigHB}) via a combination of a rate-distortion code, intended
for decoder 1, and a conditional rate-distortion code, intended for decoder 2,
which is then revealed to both the decoders. However, conditioned on what is
decodable by decoder 1, i.e., the rate-distortion code, the conditional
rate-distortion code does not explicitly ensure the asymptotic independence of
the resulting index with the side information $Y$, and therefore, does not
simplify the equivocation computation at decoder 1. To resolve this
difficulty, we present a two-step encoding scheme in which the first step is
the same as in the Kaspi problem while in the second step we first choose the
codeword intended for decoder $2$ and then bin it. We prove that the resulting
conditional bin index is asymptotically independent of the side information
$Y$.

The last part of our paper focuses on a specific source model, a binary
equiprobable source $X$ with erased side information $Y$ (with erasure
probability $p)$ and Hamming distortion constraints. For this source pair, we
focus on the rate-distortion-equivocation tradeoffs for both the uninformed
and informed cases.

For the uninformed encoder case, we prove that the maximal equivocation is
independent of the fidelity requirement $D_{2}$ at decoder $2,$ i.e., the only
information leaked about the side information is a direct consequence of the
distortion requirement at decoder $1$. We also explicitly characterize the
rate-distortion-equivocation tradeoff for this problem over the space of all
achievable distortion pairs. Our results clearly demonstrate the optimality of
the Heegard-Berger encoding scheme from both rate and equivocation standpoints.

In contrast, for the informed encoder case, we explicitly demonstrate the
usefulness of encoder side information. We first prove that the set of
distortion pairs for which perfect equivocation is achievable at decoder 1 is
strictly larger than that for the uninformed case. We prove this by showing
that the informed encoder uses the side information $Y$ via a single
description which satisfies the distortion constraints at both the decoders
while simultaneously achieving perfect privacy at decoder 1. Furthermore, we
also demonstrate that access to side information leads to a tradeoff between
rate and equivocation. To guarantee a desired equivocation, we show that the
minimal rate required can be strictly larger than the rate-distortion function
for the original Kaspi problem.%

\begin{figure}
[ptb]
\begin{center}
\includegraphics[
height=2.5079in,
width=5.1785in
]%
{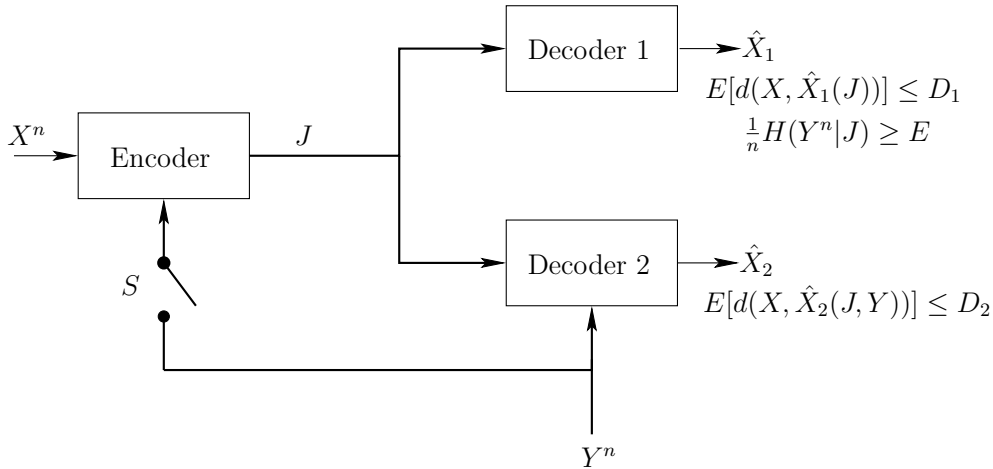}%
\caption{Source network model.}%
\label{FigHB}%
\end{center}
\end{figure}

The problem of source coding with equivocation constraints has gained
attention recently
\cite{Deniz:ISIT08,Grokop:ISIT05,JofPian,TandonISIT2011,SankarISIT2011,TandonIT,SRV4,PrabhakaranITW,Cuff,Yamamoto2Dec,MerhavArikan}%
. In contrast to these papers where the focus is on an external eavesdropper,
we address the problem of privacy leakage to a legitimate user, i.e., we seek
to understand whether the encoding at the source can discriminate between
legitimate users with and without access to correlated side information.
Furthermore, our results on the rate-distortion-equivocation tradeoff for a
binary symmetric source with erased side information for both the informed and
uninformed encoder cases allow a clear comparison of the results for the same
models without an additional privacy constraint as studied in
\cite{PerronISIT} and \cite{Perron}.

The paper is organized as follows. In Section \ref{sec:system}, we present the
system model. In Section \ref{sec:results}, we first prove the asymptotic
independence of the bin index and the decoder side information in the
Slepian-Wolf and Wyner-Ziv source coding problems. Subsequently, we establish
the rate-equivocation tradeoff regions for both the uninformed and informed
cases. In Section \ref{sec:compare}, we characterize the achievable
rate-distortion-equivocation tradeoff for a specific source pair $\left(
X,Y\right)  $ where $X$ is binary and $Y$ results from passing $X$ through an
erasure channel. We conclude in Section \ref{sec:conclusion}.

\section{\label{sec:system}System Model}

We consider a source network with a single encoder which observes and
communicates all or a part $\left(  X^{n}\right)  $ of a discrete, memoryless
bivariate source $\left(  X^{n},Y^{n}\right)  $ over a finite rate link to
decoders $1$ and $2$ at distortions $D_{1}$ and $D_{2}$, respectively, in
which decoder $2$ has access to $Y^{n}$ and an equivocation $E$ about $Y^{n}$
is required at decoder $1$. The network is shown in Fig. \ref{FigHB} where the
two cases with and without side information at the encoder correspond to the
switch $S$ being in the closed and open positions, respectively. Without the
equivocation constraint at decoder $1$, the problems with the switch in open
and closed positions, are the Heegard-Berger and Kaspi problems for which the
set of feasible $\left(  R,D_{1},D_{2}\right)  $ tuples are characterized by
Heegard and Berger \cite{HeegardBerger} and Kaspi \cite{KaspiSI},
respectively. We seek to characterize the set of all achievable $\left(
R,D_{1},D_{2},E\right)  $ tuples for both problems.

Formally, let $\left(  \mathcal{X},\mathcal{Y},p\left(  x,y\right)  \right)  $
denote the bivariate source with random variables $X\in\mathcal{X}$ and
$Y\in\mathcal{Y}$. Furthermore, let $\hat{X}_{1}$ and $\hat{X}_{2}$ denote the
reconstruction alphabets at decoders 1 and 2, respectively, and let $d_{1}$
and $d_{2}$ such that%
\begin{equation}
d_{k}:\mathcal{X}\times\mathcal{\hat{X}}\rightarrow\lbrack0,\infty),\text{
}k=1,2,
\end{equation}
be distortion measures associated with reconstruction of $X$ at decoders 1 and
2, respectively. Let $S$ take the values 0 and 1 to denote the open and closed
switch positions, respectively. An $(n,M,D_{1},D_{2},E)$ code for this network
consists of an encoder
\begin{equation}
f:\mathcal{X}^{n}\times S\cdot\mathcal{Y}^{n}\rightarrow\mathcal{J}%
=\{1,\ldots,M\}
\end{equation}
and two decoders,
\begin{align*}
g_{1} &  :\{1,\ldots,M\}\rightarrow\mathcal{\hat{X}}_{1}^{n},\text{ \ and}\\
g_{2} &  :\{1,\ldots,M\}\times\mathcal{Y}^{n}\rightarrow\mathcal{\hat{X}}%
_{2}^{n}.
\end{align*}
The expected distortion $D\,_{k}$ at decoder $k$ is given by
\begin{equation}
D_{k}=\mathbb{E}\frac{1}{n}%
{\textstyle\sum\limits_{i=1}^{n}}
d_{k}\left(  X_{i},\hat{X}_{i}\right)  ,\text{ }k=1,2,\label{Dist}%
\end{equation}
where $\hat{X}_{1}=g_{1}\left(  f\left(  X^{n}\right)  \right)  $, $\hat
{X}_{2}=g_{2}\left(  f\left(  X^{n}\right)  ,Y^{n}\right)  ,$ and the
equivocation rate $E$ is given by
\begin{equation}
E=\frac{1}{n}H\left(  Y^{n}|J\right)  ,\text{ }J\in\mathcal{J}.\label{Eq}%
\end{equation}

\begin{definition}
The rate-distortion-equivocation tuple $(R,D_{1},D_{2},E)$ is achievable for
the above source network if there exists an $(n,M,D_{1}+\epsilon
,D_{2}+\epsilon,E-\epsilon)$ code with $M\leq2^{n\left(  R+\epsilon\right)  }$
for $n$ sufficiently large. Let $\mathcal{R}$ denote the set of all achievable
$(R,D_{1},D_{2},E)$ tuples, $R\left(  D_{1},D_{2},E\right)  $ denote the
minimal achievable rate $R$, and $\Gamma\left(  D_{1},D_{2}\right)  $ denote
the maximal achievable equivocation $E$ such that
\begin{align}
R\left(  D_{1},D_{2},E\right)   &  \equiv\min_{\left(  R,D_{1},D_{2},E\right)
\in\mathcal{R}}R,\text{ \ and}\\
\Gamma\left(  D_{1},D_{2}\right)   &  \equiv\max_{\left(  R,D_{1}%
,D_{2},E\right)  \in\mathcal{R}\text{,}\forall R\geq0}E.
\end{align}

\end{definition}

\begin{remark}
$\Gamma\left(  D_{1},D_{2}\right)  $ is the maximal privacy achievable about
$Y^{n}$ at decoder 1 and $R\left(  D_{1},D_{2},E\right)  $ is the minimal rate
required to guarantee a distortion pair $\left(  D_{1},D_{2}\right)  $ and an
equivocation $E$. $R\left(  D_{1},D_{2},\Gamma(D_{1},D_{2})\right)  $ is the
minimal rate achieving the maximal equivocation for a distortion pair $\left(
D_{1},D_{2}\right)  .$
\end{remark}

\section{\label{sec:results}Related Observations}

In the context of lossless communications, \cite{Slepiansr:1973} studies a
problem of losslessly communicating a bivariate source $\left(  X,Y\right)  $
to a single decoder via two encoders, one with access to the $X^{n}$ sequences
and the other with access to the $Y^{n}$ sequences. A special case of this
problem is one in which the decoder has perfect access to $Y^{n}$ for which a
minimal rate of $R_{X}\geq H(X|Y)$ is needed \cite{Slepiansr:1973} and it this
problem (which leads to a corner point in the Slepian-Wolf region) that we
address below.

On the other hand, \cite{Wyner_Ziv} studies the problem of lossily
communicating a part $X$ of a bivariate source $\left(  X,Y\right)  $ subject
to a fidelity criterion to a single decoder which has access to $Y$ and proves
that a minimum rate of $R(D)\geq\min\left(  I(X;U)-I(Y;U)\right)  $ where the
minimization is over all distributions $p\left(  u|x\right)  $ and
deterministic functions $g$ such that $\hat{X}=g\left(  U,Y\right)  $ and
$\mathbb{E}\left[  d\left(  X,\hat{X}\right)  \right]  \leq D$.

In both of the abovementioned problems, the coding index communicated is
chosen with knowledge of the decoder side information. In the lemmas that
follow we prove that in both cases the optimal encoding is such that the
coding index is asymptotically independent of the side information $Y^{n}$ at
the decoder.

\subsection{Slepian-Wolf Coding Coding: Independence of Bin Index and Side
Information}

\begin{lemma}
\label{LemmaSW}For a bivariate source $\left(  X,Y\right)  $ where $X^{n}$ is
encoded via the encoding function $f_{SW}:\mathcal{X}^{n}\rightarrow
J\in\{1,\ldots,M_{J}\}$ while $Y^{n}$ is available only at the decoder, we
have $\lim_{n\rightarrow\infty}\left.  H\left(  Y^{n}|J\right)  \right/
n=H\left(  Y\right)  ,$ i.e., $\lim_{n\rightarrow\infty}\left.  I\left(
Y^{n};J\right)  \right/  n\rightarrow0$.
\end{lemma}

\begin{proof}
Let $\mathcal{T}_{A}\left(  n,\epsilon\right)  $ denote the set of strongly
typical $A$ sequences of length $n$. We define a binary random variable $\mu$
as follows:%
\begin{equation}
\mu\left(  x^{n},y^{n}\right)  =\left\{
\begin{array}
[c]{ll}%
0, & \left(  x^{n},y^{n}\right)  \not \in \mathcal{T}_{XY}\left(
n,\epsilon\right)  ;\\
1, & \text{otherwise.}%
\end{array}
\right.  \label{mu_def}%
\end{equation}
From the Slepian-Wolf encoding, since a typical sequence $x^{n}$ is assigned a
bin (index) $j$ at random, we have that
\begin{equation}
\Pr\left(  J=j|X^{n}=x^{n}\in\mathcal{T}_{X}\left(  n,\epsilon\right)
\right)  =\frac{1}{M_{J}} \label{PrJ}%
\end{equation}
and
\begin{equation}
\Pr\left(  J=j|\mu=1\right)  =%
{\textstyle\sum\limits_{x^{n}}}
\Pr\left(  x^{n},J=j|\mu=1\right)  \in\left(  \left(  1-\epsilon\right)
/M_{j},1/M_{J}\right)  \label{PrJSW}%
\end{equation}
where we have used the fact that for a typical set $\Pr\left(  \mathcal{T}%
_{XY}\left(  n,\epsilon\right)  \right)  \geq\left(  1-\epsilon\right)  $
\cite[chap. 2]{Cover:book}.

The conditional equivocation $H\left(  Y^{n}|J\right)  $ can be lower bounded
as%
\begin{align}
H\left(  Y^{n}|J\right)   &  \geq H\left(  Y^{n}|J,\mu\right)  \label{HJ_1}\\
&  =\Pr\left(  \mu=0\right)  H\left(  Y^{n}|J,\mu=0\right)  +\Pr\left(
\mu=1\right)  H\left(  Y^{n}|J,\mu=1\right)  \nonumber\\
&  \geq\Pr\left(  \mu=1\right)  H\left(  Y^{n}|J,\mu=1\right)  \label{HJ_3}\\
&  =\Pr\left(  \mu=1\right)
{\textstyle\sum_{j}}
\Pr\left(  j|\mu=1\right)  H\left(  Y^{n}|j,\mu=1\right)  \label{HJ_4}%
\end{align}
where (\ref{HJ_1}) follows from the fact that conditioning does not increase
entropy, and (\ref{HJ_3}) from the fact that the entropy is non-negative. The
probability $\Pr\left(  y^{n}|j,\mu=1\right)  $ can be written as
\begin{subequations}
\label{SWPrY}%
\begin{align}
&  \Pr\left(  y^{n}|j,\mu=1\right)  \nonumber\\
&  =%
{\textstyle\sum\limits_{x^{n}}}
\Pr\left(  y^{n},x^{n}|j,\mu=1\right)  \\
&  =%
{\textstyle\sum\limits_{x^{n}}}
\Pr\left(  x^{n}|j,\mu=1\right)  \Pr\left(  y^{n}|x^{n},j,\mu=1\right)
\label{HJ_5}\\
&  =%
{\textstyle\sum\limits_{x^{n}}}
\frac{\Pr\left(  x^{n},j|\mu=1\right)  }{\Pr\left(  j|\mu=1\right)  }%
\Pr\left(  y^{n}|x^{n},\mu=1\right)  \\
&  \leq2^{n\epsilon^{\prime}}%
{\textstyle\sum\limits_{x^{n}}}
\frac{\Pr\left(  x^{n}|\mu=1\right)  /M_{J}}{M_{J}}\Pr\left(  y^{n}|x^{n}%
,\mu=1\right)  \label{HJ_5a}\\
&  =2^{n\epsilon^{\prime}}%
{\textstyle\sum\limits_{x^{n}}}
\Pr\left(  x^{n}|\mu=1\right)  \Pr\left(  y^{n}|x^{n},\mu=1\right)  \\
&  =2^{n\epsilon^{\prime}}\Pr\left(  y^{n}|\mu=1\right)  \\
&  \leq2^{-n\left(  H(Y)-\epsilon^{\prime\prime}\right)  }\label{HJ_Prend}%
\end{align}
where (\ref{HJ_5}) follows from (\ref{PrJ}) and the fact that $Y^{n}-X^{n}-J$
forms a Markov chain (by construction), and (\ref{HJ_5a}) follows from
(\ref{PrJSW}). Expanding $H\left(  Y^{n}|j,\mu=1\right)  $, we have%
\end{subequations}
\begin{align}
H\left(  Y^{n}|j,\mu=1\right)   &  =%
{\textstyle\sum\limits_{y^{n}}}
p\left(  y^{n}|j,\mu=1\right)  \log\frac{1}{\Pr\left(  y^{n}|j,\mu=1\right)
}\\
&  \geq%
{\textstyle\sum\limits_{y^{n}}}
p\left(  y^{n}|j,\mu=1\right)  \log2^{n\left(  H(Y)-\epsilon^{\prime\prime
}\right)  }\label{SWHY_2}\\
&  =n\left(  H(Y)-\epsilon^{\prime\prime}\right)
{\textstyle\sum\limits_{y^{n}}}
p\left(  y^{n}|j,\mu=1\right)  \\
&  \geq n\left(  1-\epsilon\right)  \left(  H(Y)-\epsilon^{\prime\prime
}\right)  \label{SWHY_4}%
\end{align}
where (\ref{SWHY_2}) results from the upper bound on $\Pr\left(  y^{n}%
|j,\mu=1\right)  $ in (\ref{HJ_Prend}) and (\ref{SWHY_4}) from the fact that
for a typical set $\Pr\left(  \mathcal{T}_{XY}\left(  n,\epsilon\right)
\right)  \geq\left(  1-\epsilon\right)  $ \cite[chap. 2]{Cover:book}. Thus,
the equivocation $H\left(  Y^{n}|J\right)  $ can be lower bounded as
\begin{align}
H\left(  Y^{n}|J\right)   &  \geq\Pr\left(  \mu=1\right)
{\textstyle\sum_{j}}
\Pr\left(  j|\mu=1\right)  \left(  1-\epsilon\right)  n\left(  H\left(
Y\right)  -\epsilon^{\prime\prime}\right)  \label{SWPrYB}\\
&  \geq n\left(  1-\epsilon\right)  ^{3}\left(  H\left(  Y\right)
-\epsilon^{\prime\prime}\right)
\end{align}
where we have used (\ref{PrJSW}) and the fact that for a typical set
$\Pr\left(  \mathcal{T}_{XY}\left(  n,\epsilon\right)  \right)  \geq\left(
1-\epsilon\right)  $ \cite[chap. 2]{Cover:book}. The proof concludes by
observing that $H\left(  Y^{n}\right)  \geq H\left(  Y^{n}|J\right)  $ and
$\epsilon\rightarrow0,$ $\epsilon^{\prime\prime}\rightarrow0$ as
$n\rightarrow\infty$.
\end{proof}

\begin{remark}
Lemma \ref{LemmaSW} captures the intuition that it suffices to encode only
that part of $X^{n}$ that is independent of the decoder side-information
$Y^{n}.$
\end{remark}

\begin{remark}
The proof of Lemma \ref{LemmaSW} does not depend on the precise bound on the
total number, $M_{J}$, of encoding indices, i.e., it holds for all choices of
$M_{J}$. In fact, the bound on $M_{J}$ is a consequence of the decoding requirements.
\end{remark}

\subsection{Wyner-Ziv Coding: Independence of Bin Index and Side Information}

\begin{lemma}
\label{LemmaWZ}For a bivariate source $\left(  X,Y\right)  $ where $X^{n}$ is
encoded via the encoding function $f_{WZ}:\mathcal{X}^{n}\rightarrow
J\in\{1,\ldots,M_{J}\}$ while $Y^{n}$ is available only at the decoder, we
have $\lim_{n\rightarrow\infty}\left.  H\left(  Y^{n}|J\right)  \right/
n=H\left(  Y\right)  ,$ i.e., $\lim_{n\rightarrow\infty}\left.  I\left(
Y^{n};J\right)  \right/  n\rightarrow0$.
\end{lemma}

\begin{proof}
Let $T_{A}\left(  n,\epsilon\right)  $ denote the set of strongly typical $A$
sequences of length $n$. We define a binary random variable $\mu$ as follows:%
\begin{equation}
\mu\left(  u^{n},y^{n}\right)  =\left\{
\begin{array}
[c]{ll}%
0, & \left(  u^{n},y^{n}\right)  \not \in \mathcal{T}_{UY}\left(
n,\epsilon\right)  ;\\
1, & \text{otherwise.}%
\end{array}
\right.
\end{equation}
From the Wyner-Ziv encoding, for a given $x^{n}$, first a sequence $u^{n}$
that is jointly typical with $x^{n}$ is chosen where the $n$ symbols of
$u^{n}$ are generated independently according to $p_{U}\left(  \cdot\right)  $
(computed from $p_{XU}\left(  \cdot\right)  )$. The resulting sequence $u^{n}$
is assigned a bin (index) $j$ at random such that we have
\begin{equation}
\Pr\left(  J=j|U^{n}=u^{n}\in\mathcal{T}_{U}\left(  n,\epsilon\right)
\right)  =\frac{1}{M_{J}}\label{PrJWZ}%
\end{equation}
and
\begin{equation}
\Pr\left(  J=j|\mu=1\right)  =%
{\textstyle\sum\limits_{u^{n}}}
\Pr\left(  u^{n},J=j|\mu=1\right)  \in\left(  \left(  1-\epsilon\right)
/M_{j},1/M_{J}\right)  \label{PrJWZa}%
\end{equation}
where we have used the fact that the probability of the typical set
$\mathcal{T}_{UY}\left(  n,\epsilon\right)  \geq\left(  1-\epsilon\right)  $
\cite[chap. 2]{Cover:book} and using $\left(  1-\epsilon\right)
/M_{j}=2^{-n\epsilon^{\prime}}/M_{j}$ for a given $n$.

The conditional equivocation $H\left(  Y^{n}|J\right)  $ can be lower bounded
as%
\begin{align}
H\left(  Y^{n}|J\right)   &  \geq H\left(  Y^{n}|J,\mu\right)  \label{WZHJ1}\\
&  =\Pr\left(  \mu=0\right)  H\left(  Y^{n}|J,\mu=0\right)  +\Pr\left(
\mu=1\right)  H\left(  Y^{n}|J,\mu=1\right)  \nonumber\\
&  \geq\Pr\left(  \mu=1\right)  H\left(  Y^{n}|J,\mu=1\right)  \label{WZHJ3}\\
&  =\Pr\left(  \mu=1\right)
{\textstyle\sum_{j}}
\Pr\left(  j|\mu=1\right)  H\left(  Y^{n}|j,\mu=1\right)
\end{align}
where (\ref{HJ_1}) follows from the fact that conditioning reduces entropy,
and (\ref{HJ_3}) from the fact that the entropy is non-negative. The
probability $\Pr\left(  y^{n}|j,\mu=1\right)  $ can be written as
\begin{subequations}
\begin{align}
&  \Pr\left(  y^{n}|j,\mu=1\right)  \\
&  =%
{\textstyle\sum\limits_{u^{n}}}
\Pr\left(  y^{n},u^{n}|j,\mu=1\right)  \\
&  =%
{\textstyle\sum\limits_{u^{n}}}
\Pr\left(  u^{n}|j,\mu=1\right)  \Pr\left(  y^{n}|u^{n},j,\mu=1\right)  \\
&  =%
{\textstyle\sum\limits_{u^{n}}}
\Pr\left(  u^{n}|j,\mu=1\right)  \Pr\left(  y^{n}|u^{n},\mu=1\right)
\label{PrYWZ1}\\
&  =%
{\textstyle\sum\limits_{u^{n}}}
\frac{\Pr\left(  u^{n}|\mu=1\right)  }{\Pr\left(  j|\mu=1\right)  }\frac
{1}{M_{J}}\Pr\left(  y^{n}|u^{n},\mu=1\right)  \\
&  \leq%
{\textstyle\sum\limits_{u^{n}}}
\Pr\left(  u^{n}|\mu=1\right)  2^{n\epsilon^{\prime}}\Pr\left(  y^{n}%
|u^{n},\mu=1\right)  \label{PrYWZ3}\\
&  =%
{\textstyle\sum\limits_{u^{n}}}
\Pr\left(  y^{n},u^{n}|\mu=1\right)  2^{n\epsilon^{\prime}}\\
&  =\Pr\left(  y^{n}|\mu=1\right)  2^{n\epsilon^{\prime}}\\
&  \leq2^{-n\left(  H\left(  Y\right)  -\epsilon^{\prime\prime}\right)
}\label{PrYWZend}%
\end{align}
where (\ref{PrYWZ1}) follows from (\ref{PrJWZ}) and the fact that $Y^{n}%
-U^{n}-J$ forms a Markov chain (by construction) and (\ref{PrYWZ3}) follows
from (\ref{PrJWZa}). Expanding $H\left(  Y^{n}|j,\mu=1\right)  $, we have%
\end{subequations}
\begin{align}
H\left(  Y^{n}|j,\mu=1\right)   &  =%
{\textstyle\sum\limits_{y^{n}}}
p\left(  y^{n}|j,\mu=1\right)  \log\frac{1}{\Pr\left(  y^{n}|j,\mu=1\right)
}\\
&  \geq%
{\textstyle\sum\limits_{y^{n}}}
p\left(  y^{n}|j,\mu=1\right)  \log2^{n\left(  H(Y)-\epsilon^{\prime}\right)
}\\
&  =n\left(  H(Y)-\epsilon^{\prime}\right)
{\textstyle\sum\limits_{y^{n}}}
p\left(  y^{n}|j,\mu=1\right)  \\
&  \geq n\left(  1-\epsilon\right)  \left(  H(Y)-\epsilon^{\prime}\right)
\end{align}
where (\ref{SWHY_2}) results from the upper bound on $\Pr\left(  y^{n}%
|j,\mu=1\right)  $ in (\ref{PrYWZend}) and (\ref{SWHY_4}) from the fact that
for a typical set $T_{XY}\left(  n,\epsilon\right)  \geq\left(  1-\epsilon
\right)  $ \cite[chap. 2]{Cover:book}. Thus, the equivocation $H\left(
Y^{n}|J\right)  $ can be lower bounded as
\begin{align}
H\left(  Y^{n}|J\right)   &  \geq\Pr\left(  \mu=1\right)
{\textstyle\sum_{j}}
\Pr\left(  j|\mu=1\right)  \left(  1-\epsilon\right)  n\left(  H\left(
Y\right)  -\epsilon^{\prime}\right)  \\
&  \geq n\left(  1-\epsilon\right)  ^{3}\left(  H\left(  Y\right)
-\epsilon^{\prime}\right)
\end{align}
where we have used the fact that for a typical set $\mathcal{T}_{UY}\left(
n,\epsilon\right)  \geq\left(  1-\epsilon\right)  $ \cite[chap. 2]%
{Cover:book}. The proof concludes by observing that $H\left(  Y^{n}\right)
\geq H\left(  Y^{n}|J\right)  $ and $\epsilon^{\prime}\rightarrow0,$
$\epsilon^{\prime\prime}\rightarrow0$ as $n\rightarrow\infty$.
\end{proof}

We will now use Lemmas \ref{LemmaSW} and \ref{LemmaWZ} to demonstrate the
optimality of the Heegard-Berger and Kaspi encoding for the uninformed and
informed source models respectively.

\subsection{Uninformed Encoder with Side Information\ Privacy}

We first consider the source network in which the encoder does not have side
information and derive the set of all feasible rate-distortion-equivocation
(RDE) pairs. The resulting problem may be viewed as the Heegard-Berger problem
with an additional privacy constraint at decoder $1.$ Our result demonstrates
that the optimal coding scheme is the same as the Heegard-Berger problem
without a privacy constraint. The proof makes use of the independence of the
Wyner-Ziv binning index from the side information $Y^{n}$ in tightly bounding
the achievable equivocation. We briefly sketch the proof here; the detailed
proof can be found in the appendix.

\subsubsection{Rate-Distortion-Equivocation $\left(  R,D_{1},D_{2},E\right)  $
Tuples}

\begin{definition}
\label{DefRDEHB}Let $\Gamma_{U}(D_{1},D_{2})$ and $R_{U}\left(  D_{1}%
,D_{2},E\right)  $ be two functions defined as
\begin{align}
\Gamma_{U}\left(  D_{1},D_{2}\right)   &  \equiv\max_{\mathcal{P}_{U}%
(D_{1},D_{2},E)}H(Y|W_{1})\label{Gamma_RDE}\\
R_{U}\left(  D_{1},D_{2},E\right)   &  \equiv\min_{\mathcal{P}_{U}(D_{1}%
,D_{2},E)}I(X;W_{1})+I(X;W_{2}|W_{1}Y)\label{Rate_RDE}%
\end{align}
such that
\begin{equation}
\mathcal{R}_{U}\equiv\left\{  \left(  R,D_{1},D_{2},E\right)  :D_{1}%
\geq0,D_{2}\geq0,0\leq E\leq\Gamma_{U}\left(  D_{1},D_{2}\right)  ,R\geq
R_{U}\left(  D_{1},D_{2},E\right)  \right\}
\end{equation}
where the subscript $U$ denotes the uninformed case, $\mathcal{P}_{U}\left(
D_{1},D_{2},E\right)  $ is the set of all $p\left(  x,y)p(w_{1},w_{2}%
|x\right)  $ that satisfy (\ref{Dist}) and (\ref{Eq}), $Y-X-\left(
W_{1},W_{2}\right)  $ is a Markov chain, and $\left\vert \mathcal{W}%
_{1}\right\vert =\left\vert \mathcal{X}\right\vert +2,$ $\left\vert
\mathcal{W}_{2}\right\vert =\left(  \left\vert \mathcal{X}\right\vert
+1\right)  ^{2}$.
\end{definition}

\begin{lemma}
\label{Lemma_TD}$\Gamma_{U}\left(  D_{1},D_{2}\right)  $ is a non-decreasing,
concave function of $\left(  D_{1},D_{2}\right)  $ (i.e., for all $D_{l}\geq
0$, \thinspace$l=1,2$)$.$
\end{lemma}

Lemma \ref{Lemma_TD} follows from the concavity properties of the
(conditional) entropy function as a function of the underlying distribution,
and therefore, of the distortion.

\begin{theorem}
\label{ThHBRDE}For a bivariate source $\left(  X,Y\right)  $ where only
$X^{n}$ is available at the source, and $Y^{n}$ is available at decoder 2 but
not at decoder 1, we have%
\begin{equation}%
\begin{array}
[c]{ccc}%
\mathcal{R}=\mathcal{R}_{U}, & \Gamma\left(  D_{1},D_{2}\right)  =\Gamma
_{U}\left(  D_{1},D_{2}\right)  \text{,} & \text{and }R\left(  D_{1}%
,D_{2},E\right)  =R_{U}\left(  D_{1},D_{2},E\right)  .
\end{array}
\end{equation}

\end{theorem}

\textit{Proof sketch}: \textit{Converse}: A lower bound on $R\left(
D_{1},D_{2},E\right)  $ is the same as that in \cite{HeegardBerger} and
involves the introduction of two auxiliary variables $W_{1,i}\equiv\left(
J,Y^{i-1}\right)  $ and $W_{2,i}\equiv\left(  X^{i-1}Y_{i+1}^{n}\right)  $.
Using this definition of $W_{1,i}$, one can expand the equivocation definition
in (\ref{Eq}) to show that $\Gamma(D_{1},D_{2})\leq H(Y|W_{1})$.

\textit{Achievable scheme}: The achievable scheme begins with a
rate-distortion code for decoder 1 by mapping an observed $x^{n}$ sequence to
one of a set of $2^{nI(X;W_{1})}$ $w_{1}^{n}$ sequences, denoted $w_{1}%
^{n}\left(  j_{1}\right)  $, subject to typicality requirements. For this
choice of $w_{1}^{n}\left(  j_{1}\right)  $, a second code for decoder 2
results from choosing a conditionally typical sequence out of a set of
$2^{nI(X;W_{2}|W_{1})}$ $w_{2}^{n}$ sequences, denoted by $w_{2}^{n}\left(
j_{2}|j_{1}\right)  $, and binning the resulting sequence into one of
$2^{n\left(  I(X_{1};W_{2}|W_{1})-I(Y;W_{2}|W_{1})\right)  }$ bins, denoted by
$b\left(  j_{2}\right)  $, chosen uniformly. The pair $\left(  j_{1},b\left(
j_{2}\right)  \right)  $ is revealed to the decoders. We show in the appendix
that this scheme achieves an equivocation of $H\left(  Y|W_{1}\right)  $
asymptotically; the crux of our proof relies on the fact that the binning
index $B\left(  J_{2}\right)  $ is conditionally independent of $\left(
XW_{1}\right)  $ conditioned on $W_{2}$, i.e., the random variables are
related via the Markov chain relationship $Y-\left(  XW_{1}\right)
-W_{2}-B\left(  J_{2}\right)  $.

\begin{remark}
An intuitive way to interpret the equivocation arises from the following
decomposition:%
\begin{subequations}
\label{HAchInt}%
\begin{align}
\frac{1}{n}H(Y^{n}|J_{1},B\left(  J_{1},J_{2}\right)  ) &  =\frac{1}{n}%
H(Y^{n}|J_{1})\\
&  \text{ \ \ \ }-\frac{1}{n}I\left(  Y^{n};B\left(  J_{2}\right)
|J_{1}\right)  \\
&  =\frac{1}{n}H(Y^{n}|W_{1}^{n}\left(  J_{1}\right)  )\label{RemEq}\\
&  \text{ \ \ \ }-\frac{1}{n}I\left(  Y^{n};B\left(  J_{2}\right)  |W_{1}%
^{n}\left(  J_{1}\right)  \right)  .\nonumber
\end{align}
The first term in (\ref{RemEq}) is approximately equal to $H\left(
Y|W_{1}\right)  $ while the second term, which in the limit goes to $0$,
follows from a conditional version of Lemma \ref{LemmaWZ} and the fact that
$Y-X-\left(  W_{1}W_{2}\right)  -B\left(  J_{2}\right)  $ forms a Markov chain.
\end{subequations}
\end{remark}

\subsection{Informed Encoder with Side Information\ Privacy}

We now consider the source network in which the encoder has access to the side
information $Y^{n}$ and derive the set of all feasible
rate-distortion-equivocation tuples. The resulting problem may be viewed as
the Kaspi problem with an additional privacy constraint about $Y^{n}$ at
decoder $1.$ Our results below demonstrate that the Kaspi coding scheme
achieves the set of all rate-distortion-equivocation tuples. However, for a
given $\left(  D_{1},D_{2},E\right)  $ pair, the minimal rate $R(D_{1}%
,D_{2},E)$ will in general be different from the $R(D_{1},D_{2})$ for the
original Kaspi problem.

Our proof includes a two-step achievable scheme involving binning for the
conditional rate-distortion function for which we show that the bin index is
independent of the side information $Y^{n}$. Our converse is a minor
modification of the converse in \cite{KaspiSI} and involves two auxiliary
random variables. We briefly sketch the proof here; the details are relegated
to the appendix.

\subsubsection{Rate-Distortion-Equivocation $\left(  R,D_{1},D_{2},E\right)  $
Tuples}

\begin{definition}
\label{DefRDEKaspi}Let $\Gamma_{I}(D_{1},D_{2})$ and $R_{I}\left(  D_{1}%
,D_{2},E\right)  $ be two functions defined as
\begin{align}
\Gamma_{I}\left(  D_{1},D_{2}\right)   &  \equiv\max_{\mathcal{P}_{I}%
(D_{1},D_{2},E)}H(Y|W_{1})\text{, and}\label{GammaRDE_I}\\
R_{I}\left(  D_{1},D_{2},E\right)   &  \equiv\min_{\mathcal{P}_{I}(D_{1}%
,D_{2},E)}I(XY;W_{1})+I(X;W_{2}|W_{1}Y)
\end{align}
such that
\begin{equation}
\mathcal{R}_{I}\equiv\left\{  \left(  R,D_{1},D_{2},E\right)  :D_{1}%
\geq0,D_{2}\geq0,0\leq E\leq\Gamma_{I}\left(  D_{1},D_{2}\right)  ,R\geq
R_{I}\left(  D_{1},D_{2},E\right)  \right\}
\end{equation}
where $\mathcal{P}_{I}\left(  D_{1},D_{2},E\right)  $ is the set of all
$p\left(  x,y)p(w_{1},w_{2}|x,y\right)  $ that satisfy (\ref{Dist}) and
(\ref{Eq}) and $\left\vert \mathcal{W}_{1}\right\vert =\left\vert
\mathcal{X}\right\vert +2,$ $\left\vert \mathcal{W}_{2}\right\vert =\left(
\left\vert \mathcal{X}\right\vert +1\right)  ^{2}$.
\end{definition}

\begin{remark}
The cardinality bounds on $\mathcal{W}_{1}$ and $\mathcal{W}_{2}$ can be
obtained analogously to the arguments in \cite[p. 730]{HeegardBerger}.
\end{remark}

\begin{lemma}
$R_{I}\left(  D_{1},D_{2},E\right)  $ is a convex function of $\left(
D_{1},D_{2},E\right)  $.
\end{lemma}

\begin{theorem}
\label{Th_KasRDE}For a two-source $\left(  X,Y\right)  $ where $X^{n}$ is
available at the source, and $Y^{n}$ is available at the source and at decoder
2 but not at decoder 1, we have%
\begin{equation}%
\begin{array}
[c]{ccc}%
\mathcal{R}=\mathcal{R}_{I}, & \Gamma\left(  D_{1},D_{2}\right)  =\Gamma
_{I}\left(  D_{1},D_{2}\right)  \text{,} & \text{and }R\left(  D_{1}%
,D_{2},E\right)  =R_{I}\left(  D_{1},D_{2},E\right)  .
\end{array}
\end{equation}

\end{theorem}

\textit{Proof sketch}: \textit{Converse}: A lower bound on $R\left(
D_{1},D_{2},E\right)  $ can be obtained analogously to the bounds in
\cite{KaspiSI} with the introduction of two auxiliary variables $W_{1,i}%
\equiv\left(  J,Y^{i-1}\right)  $ and $W_{2,i}\equiv\left(  X^{i-1}Y_{i+1}%
^{n}\right)  $. Using this definition of $W_{1,i}$, one can expand the
equivocation definition in (\ref{Eq}) to obtain $\Gamma(D_{1},D_{2})\leq
H(Y|W_{1})$.

\textit{Achievable scheme}: The achievable scheme begins with a
rate-distortion code for decoder 1 by mapping an observed $\left(  x^{n}%
,y^{n}\right)  $ sequence to one of a set of $2^{nI(XY;W_{1})}$ $w_{1}^{n}$
sequences, denoted by $w_{1}^{n}\left(  j_{1}\right)  $, subject to typicality
requirements. A second rate-distortion code for decoder 2 results from mapping
$\left(  x^{n},y^{n},w_{1}^{n}\right)  $ to one of a set of $2^{nI(XYW_{1}%
;W_{2})}$ $w_{2}^{n}$ sequences, denoted by $w_{2}^{n}\left(  j_{2}\right)  $,
and binning the resulting sequence into one of $2^{n\left(  I(XYW_{1}%
;W_{2})-I(YW_{1};W_{2})\right)  }$ bins, denoted by $b\left(  j_{2}\right)  $,
chosen uniformly. The pair $\left(  j_{1},b\left(  j_{2}\right)  \right)  $ is
revealed to the decoders. In the appendix it is shown that this scheme
achieves an equivocation of $H\left(  Y|W_{1}\right)  $; the crux of the proof
relies on the fact that the binning index $B\left(  J_{2}\right)  $ is
conditionally independent of $\left(  XYW_{1}\right)  $ conditioned on $W_{2}$.

\begin{remark}
An intuitive way to interpret the equivocation arises from the same
decomposition as in (\ref{HAchInt}) where the first term in (\ref{RemEq}) is
approximately equal to $H\left(  Y|W_{1}\right)  $ while the second term,
which in the limit goes to $0$, follows from a conditional version of Lemma
\ref{LemmaWZ}. Note that, in contrast to the uninformed case, the distribution
here is such that $\left(  XY\right)  -\left(  W_{1}W_{2}\right)  -B\left(
J_{2}\right)  $ forms a Markov chain.
\end{remark}

\section{\label{sec:compare}Results for a Binary Source with Erased Side
Information}

We consider the following pair of correlated sources. $X$ is binary and
uniform, and
\[
Y=\left\{
\begin{array}
[c]{ll}%
X, & \hbox{w.p. }(1-p)\\
E, & \hbox{w.p. }p,
\end{array}
\right.
\]
and we consider the Hamming distortion metric, i.e., $d(x,\hat{x})=x\oplus
\hat{x}$ for both decoders and for both the informed and uninformed cases.

\subsection{Uninformed Case}

We are interested in the rate-distortion-equivocation tradeoff, given as,
\begin{align}
R &  \geq I(X;W_{1})+I(X;W_{2}|Y,W_{1}),\text{ and}\\
E &  \leq H(Y|W_{1})
\end{align}
where the rate and equivocation computation is over all random variables
$\left(  W_{1},W_{2}\right)  $ that satisfy the Markov chain relationship
$\left(  W_{1},W_{2}\right)  -X-Y$ and for which there exist functions
$f_{1}\left(  \cdot\right)  $ and $f_{1}\left(  \cdot,\cdot,\cdot\right)  $
satisfying%
\begin{align}
E[d(X,f_{1}(W_{1}))] &  \leq D_{1}\text{, and}\\
E[d(X,f_{2}(W_{1},W_{2},Y))] &  \leq D_{2}.
\end{align}
Let $h\left(  a\right)  $ denote the binary entropy function defined for
$a\in\lbrack0,1].$ The $\left(  D_{1},D_{2}\right)  $ region for this case is
partitioned into four regimes as shown in Fig. \ref{partitionHeegardBerger}.%

\begin{figure}
[ptb]
\begin{center}
\includegraphics[
height=2.9343in,
width=2.8115in
]%
{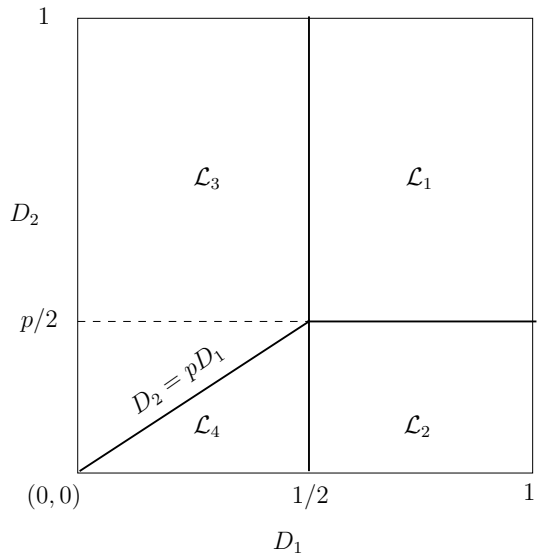}%
\caption{Partition of the $(D_{1},D_{2})$ region: uninformed encoder case.}%
\label{partitionHeegardBerger}%
\end{center}
\end{figure}

The rate-distortion-equivocation tradeoff is given as follows:
\[
R(D_{1},D_{2})=%
\begin{cases}
0; & \hbox{ if }(D_{1},D_{2})\in\mathcal{L}_{1},\\
p(1-h(D_{2}/p)); & \hbox{ if }(D_{1},D_{2})\in\mathcal{L}_{2},\\
1-h(D_{1}); & \hbox{ if }(D_{1},D_{2})\in\mathcal{L}_{3},\\
p(1-h(D_{2}/p))+(1-p)(1-h(D_{1})); & \hbox{ if }(D_{1},D_{2})\in
\mathcal{L}_{4}.
\end{cases}
\]
and
\[
\Gamma(D_{1},D_{2})=%
\begin{cases}
h(p)+(1-p)h(D_{1}); & \hbox{ if }D_{1}\leq1/2,\\
h(p)+(1-p); & \hbox{ otherwise}.
\end{cases}
\]

In Figure \ref{FigBSES}, we have plotted $R(D_{1},D_{2})$ and $\Gamma
(D_{1},D_{2})$ for the cases in which $D_{2}=p/2$ and $D_{2}=p/8$, and
$D_{1}\in\lbrack0,1/2]$.

\begin{remark}
This example shows that the equivocation does not depend on the distortion
achieved by the decoder 2 which has access to side-information $Y$, but rather
depends only on the distortion achieved by the uninformed decoder 1.
\end{remark}%

\begin{figure}
[ptb]
\begin{center}
\includegraphics[
trim=0.381334in 0.151667in 0.381916in 0.152541in,
height=2.6991in,
width=3.3503in
]%
{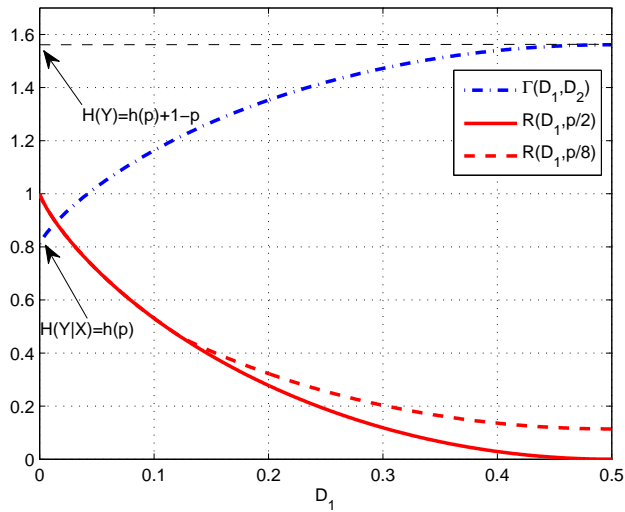}%
\caption{Illustration of the rate-equivocation tradeoff for $p=0.25$.}%
\label{FigBSES}%
\end{center}
\end{figure}

\subsubsection{Upper bound on $\Gamma(D_{1},D_{2})$}

For any $D_{1}\geq1/2$, we use the trivial upper bound
\begin{align}
\Gamma(D_{1},D_{2})  &  \leq H(Y|W_{1})\leq H(Y)\\
&  =h(p)+1-p.
\end{align}
For any $D_{1}\leq1/2$, we use the following:
\begin{subequations}
\begin{align}
\Gamma(D_{1},D_{2})  &  \leq H(Y|W_{1})\\
&  =H(Y,X|W_{1})-H(X|Y,W_{1})\\
&  =H(X|W_{1})+H(Y|X)-H(X|Y,W_{1})\\
&  =H(X|W_{1})+H(Y|X)-pH(X|W_{1})\label{EgTD4}\\
&  =H(Y|X)+(1-p)H(X|W_{1})\\
&  =H(Y|X)+(1-p)H(X|W_{1},\hat{X}_{1})\\
&  \leq H(Y|X)+(1-p)H(X|\hat{X}_{1})\\
&  \leq H(Y|X)+(1-p)H(X\oplus\hat{X}_{1})\\
&  =H(Y|X)+(1-p)h(P(X\neq\hat{X}_{1}))\\
&  \leq h(p)+(1-p)h(D_{1})
\end{align}
where (\ref{EgTD4}) follows from a direct verification that $H(X|Y,W_{1}%
)=pH\left(  X|W_{1}\right)  $ if $X$ is uniform and $Y$ is an erased version
of $X$ and $W_{1}-X-Y$ forms a Markov chain.

\subsubsection{Lower bound on $R(D_{1},D_{2})$}
\end{subequations}
\begin{itemize}
\item If $(D_{1},D_{2})\in\mathcal{L}_{1}$, we use the lower bound
$R(D_{1},D_{2})\geq0$.

\item If $(D_{1},D_{2})\in\mathcal{L}_{2}$, we use the lower bound
$R(D_{1},D_{2})\geq R_{WZ}^{(Y)}(D_{2})$ \cite{WeissmanVerdu}.

\item If $(D_{1},D_{2})\in\mathcal{L}_{3}$, we use the lower bound
$R(D_{1},D_{2})\geq1-h(D_{1})$.

\item If $(D_{1},D_{2})\in\mathcal{L}_{4}$, we show that
\begin{align}
R(D_{1},D_{2})\geq p(1-h(D_{2}/p))+(1-p)(1-h(D_{1})).
\end{align}

\end{itemize}

Consider an arbitrary $(W_{1},W_{2})$ such that $(W_{1},W_{2})\rightarrow
X\rightarrow Y$ is a Markov chain and there exist functions $f_{1}$ and
$f_{2}$:
\[
\hat{X}_{1}=f_{1}(W_{1}),\quad\text{and }\hat{X}_{2}=f_{2}(W_{1},W_{2},Y),
\]
such that
\[
\Pr(X\neq\hat{X}_{j})\leq D_{j},\quad j=1,2.
\]
Now consider the following sequence of equalities:
\begin{subequations}
\begin{align}
I(X;W_{1})+I(X;W_{2}|Y,W_{1}) &  =H(X)-H(X|W_{1})+H(X|Y,W_{1})-H(X|Y,W_{1}%
,W_{2})\nonumber\\
&  =H(X)-I(X;Y|W_{1})-H(X|Y,W_{1},W_{2})\nonumber\\
&  =H(X)-H(Y|W_{1})+H(Y|X,W_{1})-H(X|Y,W_{1},W_{2})\nonumber\\
&  =H(X)+H(Y|X)-H(Y|W_{1})-H(X|Y,W_{1},W_{2}).\label{exp1}%
\end{align}

Consider the following term appearing in (\ref{exp1}):
\end{subequations}
\begin{subequations}
\begin{align}
H(Y|W_{1}) &  =H(Y,X|W_{1})-H(X|Y,W_{1})\\
&  =H(Y|X)+H(X|W_{1})-H(X|Y,W_{1})\\
&  =H(Y|X)+(1-p)H(X|W_{1})\\
&  =H(Y|X)+(1-p)H(X|W_{1},\hat{X}_{1})\\
&  \leq H(Y|X)+(1-p)H(X|\hat{X}_{1})\\
&  \leq H(Y|X)+(1-p)H(X\oplus\hat{X}_{1})\\
&  \leq H(Y|X)+(1-p)h(D_{1}).\label{exp2}%
\end{align}
We also have
\end{subequations}
\begin{subequations}
\begin{align}
D_{2} &  \geq\Pr(X\neq\hat{X}_{2})\\
&  =\Pr(Y=E)\Pr(X\neq\hat{X}_{2}|Y=E)+\Pr(Y\neq E)\Pr(X\neq\hat{X}_{2}|Y\neq
E)\\
&  \geq\Pr(Y=E)\Pr(X\neq\hat{X}_{2}|Y=E)\\
&  =p\Pr(X\neq\hat{X}_{2}|Y=E)
\end{align}
which implies that
\end{subequations}
\begin{equation}
\Pr(X\neq\hat{X}_{2}|Y=E)\leq\frac{D_{2}}{p}\leq\frac{1}{2}.\label{exp3}%
\end{equation}
Now consider the following sequence of inequalities for the last term in
(\ref{exp1}):
\begin{subequations}
\begin{align}
H(X|Y,W_{1},W_{2}) &  =H(X|Y,W_{1},W_{2},\hat{X}_{2})\\
&  \leq H(X|Y,\hat{X}_{2})\\
&  =pH(X|Y=E,\hat{X}_{2})\\
&  \leq pH(X\oplus\hat{X}_{2}|Y=E)\\
&  =ph(P(X\neq\hat{X}_{2}|Y=E))\\
&  \leq ph(D_{2}/p)\label{exp4}%
\end{align}
where (\ref{exp4}) follows from (\ref{exp3}). Using (\ref{exp2}) and
(\ref{exp4}), we can lower bound (\ref{exp1}), to arrive at
\end{subequations}
\[
R(D_{1},D_{2})\geq p(1-h(D_{2}/p))+(1-p)(1-h(D_{1})).
\]

\subsubsection{Coding Scheme}

\begin{itemize}
\item If $(D_{1},D_{2})\in\mathcal{L}_{1}$, the $(R,\Gamma)$ tradeoff is trivial.

\item If $(D_{1},D_{2})\in\mathcal{L}_{2}$, we use the following coding scheme:

In this regime, we have $D_{1}\geq1/2$, hence the encoder sets $W_{1}=\phi$,
and sends only one description $W_{2}= X\oplus N$, where $N\sim
\hbox{Ber}(D_{2}/p)$ and $N$ is independent of $X$. It can be verified that
$I(X;W_{2}|Y)= p(1-h(D_{2}/p))$. Decoder $2$ estimates $X$ by $\hat{X}_{2}$ as
follows:
\[
\hat{X}_{2}=%
\begin{cases}
Y; & \hbox{ if }Y\neq E;\\
W_{2}; & \hbox{ if }Y=E.
\end{cases}
\]
Therefore the achievable distortion at decoder $2$ is $(1-p)0+p(D_{2}%
/p)=D_{2}$.

\item If $(D_{1},D_{2})\in\mathcal{L}_{3}$, we use the following coding scheme:

The encoder sets $W_{2}=\phi$, and sends only one description $W_{1}= X\oplus
N$, where $N\sim\hbox{Ber}(D_{1})$ and $N$ is independent of $X$. It can be
verified that $I(X;W_{1})= 1-h(D_{1})$. Decoder $1$ estimates $X$ as $\hat
{X}_{1}=W_{1}$ which leads to distortion of $D_{1}$. Decoder $2$ estimates $X$
by $\hat{X}_{2}$ as follows:
\[
\hat{X}_{2}=%
\begin{cases}
Y; & \hbox{ if }Y\neq E;\\
W_{1}; & \hbox{ if }Y=E.
\end{cases}
\]
Therefore the achievable distortion at decoder $2$ is $(1-p)0+p(D_{1})=pD_{1}%
$. Hence, as long as $D_{2}\geq pD_{1}$, the fidelity requirement of decoder
$2$ is satisfied.

\item If $(D_{1},D_{2})\in\mathcal{L}_{4}$, we use the following coding scheme:

We select $W_{2}=X\oplus N_{2}$, and $W_{1}=W_{2}\oplus N_{1}$, where
$N_{2}\sim Ber(D_{2}/p)$, and $N_{1}\sim Ber(\alpha)$, where $\alpha
=(D_{1}-D_{2}/p)/(1-2D_{2}/p)$, and the random variables $N_{1}$ and $N_{2}$
are independent of each other and are also independent of $X$. At the
uninformed decoder, the estimate is created as $\hat{X}_{1}=W_{1}$, so that
the desired distortion $D_{1}$ is achieved.

At the decoder with side-information $Y$, the estimate $\hat{X}_{2}$ is
created as follows:
\[
\hat{X}_{2}=%
\begin{cases}
Y; & \hbox{ if }Y\neq E;\\
W_{2}; & \hbox{ if }Y=E.
\end{cases}
\]
Therefore the achievable distortion at this decoder is $(1-p)0+p(D_{2}%
/p)=D_{2}$. It is straightforward to check that the rate required by this
scheme matches the stated lower bound on $R(D_{1},D_{2})$, and $\Gamma
(D_{1},D_{2})=H(Y|W_{1})=h(p)+(1-p)h(D_{1})$. This completes the proof of the
achievable part.
\end{itemize}

\subsection{\label{sec:Kaspi}Informed Encoder}

For this case, the rate-distortion-equivocation tradeoff is given as
\begin{align}
R &  \geq I(X,Y;W_{1})+I(X;W_{2}|W_{1},Y),\text{ and}\\
E &  \leq H(Y|W_{1})
\end{align}
where the joint distribution of $(W_{1},W_{2})$ with $(X,Y)$ can be arbitrary.

As in the previous section, we partition the space of admissible $(D_{1}%
,D_{2})$ distortion pairs. For simplicity, we denote these partitions as
follows:
\begin{align}
\mathcal{G}_{1} &  =\{(D_{1},D_{2}):D_{1}\geq1/2,D_{2}\geq p/2\},\\
\mathcal{G}_{2} &  =\{(D_{1},D_{2}):D_{1}\geq1/2,D_{2}\leq p/2\},\\
\mathcal{G}_{3} &  =\{(D_{1},D_{2}):D_{1}\geq D_{2}+(1-p)/2,D_{2}\leq p/2\},\\
\mathcal{G}_{4} &  =\{(D_{1},D_{2}):D_{1}\leq1/2,D_{2}\geq D_{1}\},\text{
and}\\
\mathcal{G}_{5} &  =\{(D_{1},D_{2}):D_{1}\leq D_{2}+(1-p)/2,D_{2}\leq D_{1}\}.
\end{align}
These partitions are illustrated in Figure \ref{partitionKaspi}.%

\begin{figure}
[ptb]
\begin{center}
\includegraphics[
height=2.9343in,
width=2.8115in
]%
{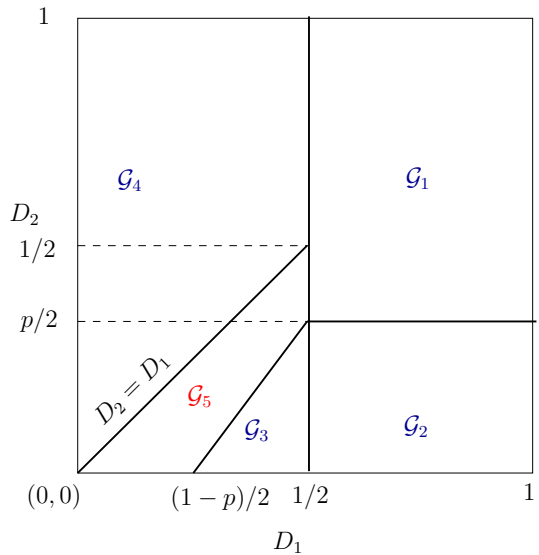}%
\caption{Partition of $(D_{1},D_{2})$ region: informed encoder case.}%
\label{partitionKaspi}%
\end{center}
\end{figure}

We provide a partial characterization the optimal $(R,E)$ tradeoff as a
function of $(D_{1},D_{2})$. In particular, we establish the tight
characterization of $(R,E)$ pairs for all values of $(D_{1},D_{2})$ with the
exception of when $(D_{1},D_{2})\in\mathcal{G}_{5}$. This characterization
reveals the benefit of the encoder side-information. It shows that in the
presence of encoder side-information, there can be several $(R,E)$ operating
points relative to the case in which the encoder does not have side-information.

(a) $(D_{1},D_{2})\in\mathcal{G}_{1}:$ In this case the $(R,\Gamma)$ region is
trivial since both the decoders can satisfy their distortion constraints which
also yields the maximum equivocation, i.e., we have
\begin{align}
R(D_{1},D_{2}) &  =0\text{, and}\\
\Gamma(D_{1},D_{2}) &  =h(p)+1-p
\end{align}

(b) $(D_{1},D_{2})\in\mathcal{G}_{2}:$ In this case, we use the proof as in
the uninformed case \ for the partition $\mathcal{L}_{2}$ to show that
\begin{align}
R(D_{1},D_{2}) &  =p(1-h(D_{2}/p)),\text{ and}\\
\Gamma(D_{1},D_{2}) &  =h(p)+1-p.
\end{align}

(c) $(D_{1},D_{2})\in\mathcal{G}_{3}:$ The $(R,\Gamma)$ tradeoff for this case
is given as follows:
\begin{align}
R(D_{1},D_{2}) &  =p(1-h(D_{2}/p)),\text{ and}\\
\Gamma(D_{1},D_{2}) &  =h(p)+1-p.
\end{align}
This case differs from the uninformed encoder case in the sense that for the
same rate, we can achieve the maximum equivocation and a non-trivial
distortion for decoder $1$. Since $R\geq R_{X|Y}(D_{2})=R_{WZ}^{Y}(D_{2})$,
and $\Gamma\leq H(Y)$, the converse proof is straightforward. The interesting
aspect of this regime is the coding scheme, which utilizes the side
information at the encoder in a non-trivial manner. To achieve this tradeoff,
we set $W_{2}=0$, and send only one description $W_{1}$ to both the decoders.
The conditional distribution $p(w_{1}|x,y)$ that is used to generate the
$W_{1}^{n}$ codewords is illustrated in Figure \ref{FigW1generation1} .%

\begin{figure}
[ptb]
\begin{center}
\includegraphics[
height=2.7916in,
width=2.7622in
]%
{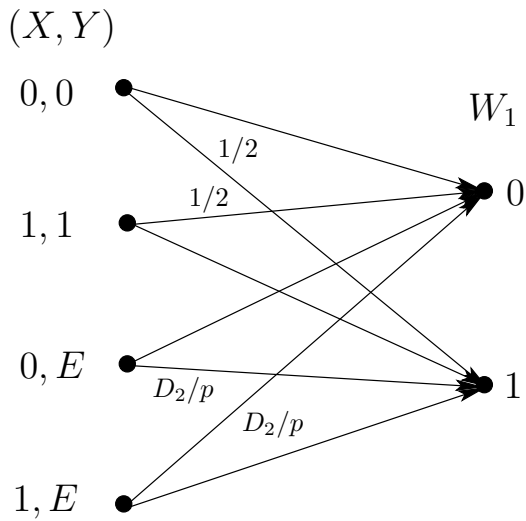}%
\caption{Illustration of $p(w_{1}|x,y)$ when $D_{1}\geq D_{2}+(1-p)/2$ and
$D_{2}\in\lbrack0,p/2]$.}%
\label{FigW1generation1}%
\end{center}
\end{figure}

Hence the rate for this scheme is given by
\begin{align}
R &  \geq I(X,Y;W_{1})\\
&  =H(W_{1})-H(W_{1}|X,Y)\\
&  =1-H(W_{1}|X,Y)\\
&  =1-(1-p)-ph(D_{2}/p)\\
&  =p(1-h(D_{2}/p)),
\end{align}
and the equivocation is given as
\begin{align}
\Gamma &  =H(Y|W_{1})\\
&  =H(Y)-I(Y;W_{1})\\
&  =H(Y)-H(W_{1})+H(W_{1}|Y)\\
&  =H(Y)-1+H(W_{1}|Y)\\
&  =H(Y)-1+(1-p)H(W_{1}|Y=X)+pH(W_{1}|Y\neq X)\\
&  =H(Y)-1+(1-p)+p\\
&  =H(Y).
\end{align}
Decoder $2$ forms its estimate as follows:
\[
\hat{X}_{2}=%
\begin{cases}
Y & \hbox{ if }Y\neq E;\\
W_{1} & \hbox{ if }Y=E,
\end{cases}
\]
which yields a distortion of $D_{2}$ at decoder $2$. Decoder $1$ forms its
estimate as
\[
\hat{X}_{1}=W_{1}%
\]
which yields
\[
\mathbb{P}(\hat{X}_{1}\neq X)=D_{2}+\frac{(1-p)}{2}.
\]
Therefore, as long as
\[
D_{1}\geq D_{2}+\frac{(1-p)}{2},
\]
this scheme achieves the optimal $(R,\Gamma)$ tradeoff.

We now informally describe the intuition behind this coding scheme: since the
encoder has access to side-information $Y$, it uses the fact that whenever
$Y=X$, no additional rate is required to satisfy the requirement of decoder
$2$, i.e., for $(1-p)$-fraction of time it is guaranteed to exactly recover
$X$. However, this yields a distortion of $(1-p)/2$ at decoder $1$ (since
decoder $1$ does not have access to $Y$). In the remaining $p$-fraction of
time, the encoder describes $X$ with a distortion $D_{2}/p$, which contributes
to a distortion of $D_{2}$ at both the decoders. To summarize, the net
distortion at decoder $2$ is $D_{2}$, whereas the distortion at decoder $1$ is
lowered from $1/2$ to $(1-p)/2+D_{2}$. Furthermore, by construction, $W_{1}$
is independent of $Y$, i.e., $H(Y|W_{1})=H(Y)$, which results in the maximal
equivocation at decoder 1.

(d) $(D_{1},D_{2})\in\mathcal{G}_{4}:$ For this case, the $(R,E)$ tradeoff is
given as the set of $(R,E)$ pairs
\begin{align}
R &  \geq1-(1-p)h\left(  \frac{D_{1}-p\alpha}{1-p}\right)  -ph(\alpha),\text{
and}\\
E &  \leq h(p)+(1-p)h\left(  \frac{D_{1}-p\alpha}{1-p}\right)  ,
\end{align}
where the parameter $\alpha$ belongs to the range $\alpha\in\lbrack0,D_{1}/p]$.

We now describe the coding scheme that achieves this region: we set
$W_{2}=\phi$, and send one description $W_{1}$ at a rate $I(X,Y;W_{1})$. The
conditional distribution $p(w_{1}|x,y)$ that is used to generate the
$W_{1}^{n}$ codewords is illustrated in Figure \ref{FigW1generation2}. The
parameters $(\alpha,\beta)$ that describe this distribution are chosen such
that
\begin{align}
D_{1} &  \geq\mathbb{P}(X\neq W_{1})\\
&  \geq(1-p)\beta+p\alpha,
\end{align}
so that $\beta\leq(D_{1}-p\alpha)/(1-p)$. At decoder $2$, the estimate
$\hat{X}_{2}$ is created as
\[
\hat{X}_{2}=%
\begin{cases}
Y; & \hbox{ if }Y\neq E;\\
W_{1}; & \hbox{ if }Y=E,
\end{cases}
\]
which yields a distortion of $p\alpha$. Since $\alpha\in\lbrack0,D_{1}/p]$,
the worst case distortion for decoder $2$ for a fixed $D_{1}$ is
$p(D_{1}/p)=D_{1}$. Hence, as long as $D_{2}\geq D_{1}$, we can satisfy the
fidelity requirements at both decoders. By direct calculations, it can be
shown that the resulting $(R,E)$ tradeoff is as stated above.

Compared to all the previous cases, the proof of optimality of the above
coding scheme is non-trivial and is relegated to the appendix.

We remark here that in this regime, the tradeoff between rate and privacy can
be observed in a precise manner. First, note that the choice $\alpha=D_{1}$
yields the $(R,E)$ operating point as in the uninformed encoder case. Next,
when $\alpha$ decreases from $D_{1}$ to $0$, the equivocation increases,
albeit at the cost of a higher rate. This phenomenon does not occur in the
case in which the encoder does not have side information.

Finally, when $\alpha$ is in the range $(D_{1},D_{1}/p]$, we obtain a lower
equivocation by increasing the rate. This phenomenon appears counterintuitive
and can be explained as follows: this range of $\alpha$ corresponds to a
coding scheme in which we give more weight to the side-information $Y$ when
describing $X$ to decoder $1$. Such a coding scheme can be regarded as the
solution to the problem in which the encoder is interested in revealing $Y$ to
decoder $1$, while simultaneously satisfying the fidelity requirement for $X$
at decoder $1$. While it is a feasible solution to the problem, it may not be
a desirable coding scheme when the privacy of $Y$ at decoder is of primary
concern, and thus, there exists a set of rate-equivocation operating points
that one can choose from. In Figure \ref{FigTradeoff2}, we show the $(R,E)$
achievable tradeoff when $p=0.4$ and $D_{1}=0.2$.%

\begin{figure}
[ptb]
\begin{center}
\includegraphics[
height=2.7899in,
width=3.1574in
]%
{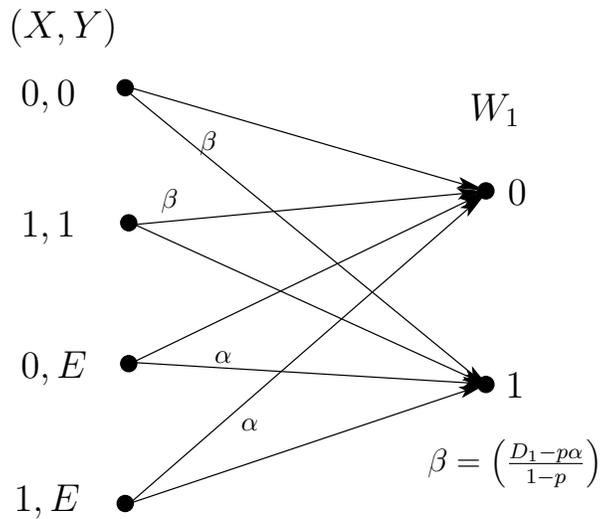}%
\caption{Illustration of $p(w_{1}|x,y)$ when $D_{1}\leq1/2,D_{2}\geq D_{1}$.}%
\label{FigW1generation2}%
\end{center}
\end{figure}
%

\begin{figure}
[ptb]
\begin{center}
\includegraphics[
height=2.9334in,
width=4.1978in
]%
{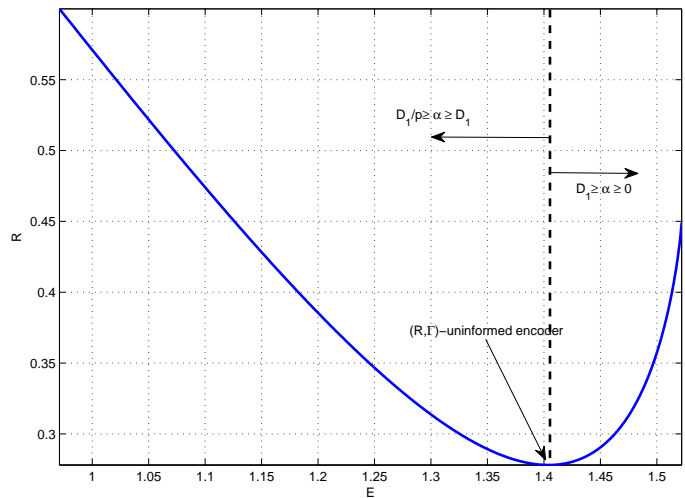}%
\caption{Illustration of the rate-equivocation tradeoff for $p=0.4,D_{1}=0.2$
with an informed encoder.}%
\label{FigTradeoff2}%
\end{center}
\end{figure}

(d) $(D_{1},D_{2})\in\mathcal{G}_{5}:$ For this case, the following $(R,E)$
pairs are achievable:
\begin{align}
R &  \geq1-(1-p)h\left(  \frac{D_{1}-p\alpha}{1-p}\right)  -ph(\alpha),\text{
and}\\
E &  \leq h(p)+(1-p)h\left(  \frac{D_{1}-p\alpha}{1-p}\right)  ,
\end{align}
where $\alpha$ is such that $\alpha\in\lbrack0,D_{2}/p]$. The coding scheme
that achieves this tradeoff is similar to the one used when $(D_{1},D_{2}%
)\in\mathcal{G}_{4}$, with the exception that the range of $\alpha$ is
different. The question of optimality of tradeoff for this regime is still unresolved.

\section{\label{sec:conclusion}Concluding Remarks}

We have determined the rate-distortion-equivocation region for a source coding
problem with two decoders, in which only one of the decoders has correlated
side information and it is desired to keep this side information private from
the uninformed decoder. We have studied two cases of this problem depending on
the availability of side information at the encoder. We have proved that the
Heegard-Berger and the Kaspi coding schemes are optimal even with an
additional privacy constraint for the uninformed and the informed encoder
cases, respectively. We have illustrated our results for a binary symmetric
source with erasure side information and Hamming distortion which clearly
highlight the difference between the informed and uninformed cases and the
advantages of encoder side information for enhancing privacy. Future work
includes generalization to multiple decoders as well as to continuously
distributed sources.%

\appendix
{}

\subsection{\label{App1}Proof of Theorem \ref{ThHBRDE}}

\textit{Converse}: The lower bound on $R\left(  D_{1},D_{2},E\right)  $ follow
directly from the converse for the Heegard-Berger problem and is omitted here
in the interest of space. We now upper bound the maximal achievable
equivocation as
\begin{subequations}
\label{HBHCon}%
\begin{align}
\frac{1}{n}H\left(  Y^{n}|J\right)   &  =%
{\textstyle\sum\limits_{i=1}^{n}}
\frac{1}{n}H\left(  Y_{i}|Y^{i-1}J\right)  \\
&  =%
{\textstyle\sum\limits_{i=1}^{n}}
\frac{1}{n}H\left(  Y_{i}|W_{i}\right)  \label{HBConTD1}\\
&  \leq\Gamma_{U}\left(  D_{1},D_{2}\right)  \label{HBConTD2}%
\end{align}
where (\ref{HBConTD1}) follows from defining $W_{1,i}\equiv\left(
J,Y^{i-1}\right)  $ (see \cite[sec. IV]{HeegardBerger}) and (\ref{HBConTD2})
follows from the definition of $\Gamma_{U}\left(  D_{1},D_{2}\right)  $ in
(\ref{Gamma_RDE}) and its concavity property from Lemma \ref{Lemma_TD}.

\textit{Achievability}: We briefly summarize the Heegard-Berger coding scheme
\cite{HeegardBerger}. Fix $p\left(  w_{1},w_{2}|x\right)  $. First generate
$M_{1}=2^{n\left(  I(W_{1};X)+\epsilon\right)  }$, $W_{1}^{n}\left(
j_{1}\right)  $ sequences, $j_{1}=1,2,\ldots,M_{1}$, independently and
identically distributed (i.i.d.) according to $p\left(  w_{1}\right)  $. For
every $W_{1}^{n}\left(  j_{1}\right)  $ sequence, generate $M_{2}=2^{n\left(
I(W_{2};X|W_{1})+\epsilon\right)  }$ $W_{2}^{n}\left(  j_{2}|j_{1}\right)  $
sequences i.i.d. according to $p\left(  w_{2}|w_{1}\left(  j_{1}\right)
\right)  $. Bin the resulting $W_{2}^{n}$ sequences into $S$ bins (analogously
to the Wyner-Ziv binning), chosen at random where $S=2^{n\left(
I(X;W_{2}|W_{1})-I(Y;W_{2}|W_{1})+\epsilon\right)  }$, and index these bins as
$b\left(  j_{2}\right)  $. Upon observing a source sequence $x^{n},$ the
encoder searches for a $W_{1}^{n}\left(  j_{1}\right)  $ sequence such that
$\left(  x^{n},w_{1}^{n}\left(  j_{1}\right)  \right)  \in\mathcal{T}_{XW_{1}%
}\left(  n,\epsilon\right)  $ (the choice of $M_{1}$ ensures that there exists
at least one such $j_{1}$). Next, the encoder searches for a $w_{2}^{n}\left(
j_{2}|j_{1}\right)  $ such that $\left(  x^{n},w_{1}^{n}\left(  j_{1}\right)
,w_{2}^{n}\left(  j_{2}|j_{1}\right)  \right)  \in\mathcal{T}_{XW_{1}W_{2}%
}\left(  n,\epsilon\right)  $ (the choice of $M_{2}$ ensures that there exists
at least one such $j_{2}$). The encoder sends $\left(  j_{1},b\left(
j_{2}\right)  \right)  $ where $b\left(  j_{2}\right)  $ is the bin index of
the $w_{2}^{n}\left(  j_{2}|j_{1}\right)  $ sequence. Thus, we have that
$\left(  XW_{1}\right)  -W_{2}-B$ forms a Markov chain and%
\end{subequations}
\begin{align}
&  \Pr\left(  B=b\left(  j_{2}\right)  |\left(  x^{n},w_{1}^{n}\left(
j_{1}\right)  ,w_{2}^{n}\left(  j_{2}|j_{1}\right)  \right)  \right.  \left.
\in\mathcal{T}_{XW_{1}W_{2}}\left(  n,\epsilon\right)  \right)  \nonumber\\
&  =\Pr\left(  B=b\left(  j_{2}\right)  |w_{2}^{n}\left(  j_{2}|j_{1}\right)
\right.  \left.  \in\mathcal{T}_{W_{2}}\left(  n,\epsilon\right)  \right)
=1/S.\label{PrBHB}%
\end{align}
With $\mu$ as defined in (\ref{mu_def}) for the typical set $\mathcal{T}%
_{XYW_{1}W_{2}}$, and $J\equiv\left(  J_{1},B\left(  J_{2}\right)  \right)  $,
the achievable equivocation can be lower bounded as
\begin{subequations}
\label{HBHAch}%
\begin{align}
&  \frac{1}{n}H\left(  Y^{n}|J_{1},B\left(  J_{2}\right)  \right)  \nonumber\\
&  \geq\frac{1}{n}H\left(  Y^{n}|J_{1},B\left(  J_{2}\right)  ,\mu\right)  \\
&  =\frac{1}{n}H\left(  Y^{n}|W_{1}^{n}\left(  J_{1}\right)  ,B\left(
J_{2}\right)  ,\mu\right)  \\
&  \geq\Pr\left(  \mu=1\right)  \frac{1}{n}H\left(  Y^{n}|W_{1}^{n}\left(
J_{1}\right)  ,B\left(  J_{2}\right)  ,\mu=1\right)  .\label{HYJ_HB}%
\end{align}
The probability $\Pr\left(  y^{n}|w_{1}^{n}\left(  j_{1}\right)  ,b\left(
j_{2}\right)  ,\mu=1\right)  $ for all $j_{1},j_{2},$ and $y^{n}$ can be
written as
\end{subequations}
\begin{subequations}
\label{HBHAch2}%
\begin{align}
&
{\textstyle\sum\limits_{\left(  x^{n},j_{2}\right)  }}
\Pr\left(  y^{n},j_{2},x^{n}|w_{1}^{n}\left(  j_{1}\right)  ,b\left(
j_{2}\right)  ,\mu=1\right)  \nonumber\\
&  =%
{\textstyle\sum\limits_{\left(  x^{n},j_{2}\right)  }}
\Pr\left(  x^{n},j_{2}|w_{1}^{n}\left(  j_{1}\right)  ,b\left(  j_{2}\right)
,\mu=1\right)  \Pr\left(  y^{n}|x^{n},\mu=1\right)  \label{HBPrY2}\\
&  =%
{\textstyle\sum\limits_{\left(  x^{n},j_{2}\right)  }}
\frac{\Pr\left(  x^{n},j_{2},w_{1}^{n}\left(  j_{1}\right)  ,b\left(
j_{2}\right)  |\mu=1\right)  }{\Pr\left(  w_{1}^{n}\left(  j_{1}\right)
,b\left(  j_{2}\right)  |\mu=1\right)  }\Pr\left(  y^{n}|x^{n},\mu=1\right)
\\
&  =%
{\textstyle\sum\limits_{\left(  x^{n},j_{2}\right)  }}
\frac{\Pr\left(  x^{n},j_{2},w_{1}^{n}\left(  j_{1}\right)  |\mu=1\right)
/S}{\Pr\left(  w_{1}^{n}\left(  j_{1}\right)  |\mu=1\right)  /S}\Pr\left(
y^{n}|x^{n},\mu=1\right)  \\
&  \leq2^{n\epsilon^{\prime}}%
{\textstyle\sum\limits_{\left(  x^{n},j_{2}\right)  }}
\Pr\left(  x^{n},j_{2}|w_{1}^{n}\left(  j_{1}\right)  ,\mu=1\right)
\Pr\left(  y^{n}|x^{n},\mu=1\right)  \label{HBPrY5}\\
&  =2^{n\epsilon^{\prime}}%
{\textstyle\sum\limits_{\left(  x^{n},j_{2}\right)  }}
\Pr\left(  x^{n},j_{2},y^{n}|w_{1}^{n}\left(  j_{1}\right)  ,\mu=1\right)  \\
&  =2^{n\epsilon^{\prime}}\Pr\left(  y^{n}|w_{1}^{n}\left(  j_{1}\right)
,\mu=1\right)
\end{align}
where (\ref{HBPrY2}) follows from the fact that $Y-X-\left(  W_{1}%
,W_{2}\right)  $ forms a\ Markov chain and (\ref{HBPrY5}) is obtained by
expanding $\Pr\left(  w_{1}^{n}\left(  j_{1}\right)  ,b\left(  j_{2}\right)
|\mu=1\right)  $ as follows:
\end{subequations}
\begin{subequations}
\label{HB_PrW1B1}%
\begin{align}
&  \Pr\left(  w_{1}^{n}\left(  j_{1}\right)  ,b\left(  j_{2}\right)
|\mu=1\right)  \nonumber\\
&  =\Pr\left(  w_{1}^{n}\left(  j_{1}\right)  |\mu=1\right)
{\textstyle\sum\limits_{w_{2}^{n}}}
\Pr\left(  b\left(  j_{2}\right)  ,w_{2}^{n}\left(  j_{1}\right)  |w_{1}%
^{n}\left(  j_{1}\right)  ,\mu=1\right)  \\
&  =\Pr\left(  w_{1}^{n}\left(  j_{1}\right)  |\mu=1\right)
{\textstyle\sum\limits_{w_{2}^{n}}}
\Pr\left(  w_{2}^{n}\left(  j_{1}\right)  |w_{1}^{n}\left(  j_{1}\right)
,\mu=1\right)  \frac{1}{S}\label{PrW1B1}\\
&  \geq\Pr\left(  w_{1}^{n}\left(  j_{1}\right)  |\mu=1\right)  \frac{\left(
1-\epsilon\right)  }{S}\label{PrW1B2}\\
&  =\Pr\left(  w_{1}^{n}\left(  j_{1}\right)  |\mu=1\right)  \frac
{2^{-n\epsilon^{\prime}}}{S}%
\end{align}
where (\ref{PrW1B1}) follows from the fact that $W_{1}-W_{2}-B$ forms a Markov
chain and (\ref{PrBHB}), while (\ref{PrW1B2}) follows the fact that for a
typical set $\Pr\left(  \mathcal{T}_{W_{1}W_{2}}\left(  n,\epsilon\right)
\right)  \geq\left(  1-\epsilon\right)  $ \cite[chap. 2]{Cover:book}. Thus,
from (\ref{HBHAch2}) we have that
\end{subequations}
\begin{align}
\Pr\left(  y^{n}|w_{1}^{n}\left(  j_{1}\right)  ,b\left(  j_{2}\right)
,\mu=1\right)   &  \leq2^{n\epsilon^{\prime}}\Pr\left(  y^{n}|w_{1}^{n}\left(
j_{1}\right)  ,\mu=1\right)  \\
&  \leq2^{-n\left(  H\left(  Y|W_{1}\right)  -\epsilon^{\prime\prime}\right)
}.\label{HBPrYW1}%
\end{align}
From (\ref{HYJ_HB}) and (\ref{HBPrYW1}), we then have%
\begin{align}
H\left(  Y^{n}|w_{1}^{n}\left(  j_{1}\right)  ,b\left(  j_{2}\right)
,\mu=1\right)   &  \geq%
{\textstyle\sum\limits_{y^{n}}}
\Pr\left(  y^{n}|w_{1}^{n}\left(  j_{1}\right)  ,\mu=1\right)  n\left(
H\left(  Y|W_{1}\right)  -\epsilon^{\prime\prime}\right)  \\
&  \geq n\left(  1-\epsilon\right)  \left(  H\left(  Y|W_{1}\right)
-\epsilon^{\prime\prime}\right)
\end{align}
such that%
\begin{align}
\frac{1}{n}H\left(  Y^{n}|J\right)   &  \geq\Pr\left(  \mu=1\right)  \frac
{1}{n}%
{\textstyle\sum\limits_{w_{1}^{n},b\left(  j_{2}\right)  }}
\Pr\left(  w_{1}^{n}\left(  j_{1}\right)  ,b\left(  j_{2}\right)
|\mu=1\right)  H\left(  Y^{n}|w_{1}^{n}\left(  j_{1}\right)  ,b\left(
j_{2}\right)  ,\mu=1\right)  \\
&  \geq\left(  1-\epsilon\right)  ^{3}\left(  H\left(  Y|W_{1}\right)
-\epsilon^{\prime\prime}\right)
\end{align}
where we have used the fact that for a typical set $\Pr\left(  \mathcal{T}%
_{YW_{1}W_{2}}\left(  n,\epsilon\right)  \right)  \geq\left(  1-\epsilon
\right)  $ \cite[chap. 2]{Cover:book}. The proof concludes by observing that
$H\left(  Y^{n}\right)  \geq H\left(  Y^{n}|J\right)  $ and $\epsilon
\rightarrow0,$ $\epsilon^{\prime\prime}\rightarrow0$ as $n\rightarrow\infty$.

\subsection{Proof of Theorem \ref{Th_KasRDE}}

\textit{Converse}: A lower bound on $R\left(  D_{1},D_{2},E\right)  $ can be
obtained as follows.
\begin{subequations}
\label{KasConR}%
\begin{align}
nR &  \geq H(J)\\
&  \geq I(X^{n}Y^{n};J)\\
&  =I(X^{n};J|Y^{n})+I(Y^{n};J)\\
&  =%
{\textstyle\sum\limits_{i=1}^{n}}
\left\{  I(X_{i};JX^{i-1}Y^{-1}Y_{i+1}^{n}|Y_{i})-I(X_{i};X^{i-1}Y^{-1}%
Y_{i+1}^{n}|Y_{i})\right.  \left.  +I(Y_{i};J,Y^{i-1})-I(Y_{i};Y^{i-1}%
)\right\}  \nonumber\\
&  =%
{\textstyle\sum\limits_{i=1}^{n}}
\left\{  I(X_{i};JX^{i-1}Y^{i-1}Y_{i+1}^{n}|Y_{i})\right.  \left.
+I(Y_{i};JY^{i-1})\right\}  \label{KasConv}\\
&  =%
{\textstyle\sum\limits_{i=1}^{n}}
\left\{  I(X_{i};JY^{i-1}|Y_{i})+I(X_{i};X^{i-1}Y_{i+1}^{n}|JY^{i-1}%
Y_{i})+I(Y_{i};JY^{i-1})\right\}  \label{KasConv2}%
\end{align}
where (\ref{KasConv}) follows from the independence of the pairs $\left(
X_{i},Y_{i}\right)  $ for all $i=1,2,\ldots,n$. Let $W_{1,i}\equiv\left(
J,Y^{i-1}\right)  $ and $W_{2,i}\equiv\left(  X^{i-1}Y_{i+1}^{n}\right)  $.
With these definitions, (\ref{KasConv2}) can be written as
\end{subequations}
\begin{align}
nR &  \geq%
{\textstyle\sum\limits_{i=1}^{n}}
\left\{  I(X_{i}Y_{i};W_{1,i})+I(X_{i};W_{2,i}|W_{1,i}Y_{i}\right\}  \\
&  \geq%
{\textstyle\sum\limits_{i=1}^{n}}
R_{I}\left(  D_{1,i},D_{2,i},E_{i}\right)  \label{ConKas3}\\
&  \geq nR_{I}\left(  D_{1},D_{2},E\right)  \label{ConKas4}%
\end{align}
where (\ref{ConKas3}) follows from Definition \ref{DefRDEKaspi} with
$D_{1,i},$ $D_{2,i},$ and $E_{i}$ defined as
\begin{subequations}
\label{KasConDE}%
\begin{align}
D_{1,i} &  \equiv\mathbb{E}\left[  d\left(  X_{i},g_{1,i}^{\prime}\left(
W_{1,i}\right)  \right)  \right]  \\
D_{2,i} &  \equiv\mathbb{E}\left[  d\left(  X_{i},g_{2,i}^{\prime}\left(
W_{2,i},W_{1,i},Y_{i}\right)  \right)  \right]  ,\text{ and}\\
E_{i} &  \equiv H(Y_{i}|W_{1,i}),
\end{align}
and (\ref{ConKas4}) follows from the convexity of $R_{I}(D_{1},D_{2},E)$ and
the definitions of $D_{k}$, \thinspace\thinspace$k=1,2,$ in (\ref{Dist}) and
the concavity of $H\left(  Y|W\right)  $, and hence, of $E$. We upper bound
the maximal achievable equivocation as
\end{subequations}
\begin{subequations}
\label{KasHYCon}%
\begin{align}
\frac{1}{n}H\left(  Y^{n}|J\right)   &  =%
{\textstyle\sum\limits_{i=1}^{n}}
\frac{1}{n}H\left(  Y_{i}|Y^{i-1}J\right)  \\
&  =%
{\textstyle\sum\limits_{i=1}^{n}}
\frac{1}{n}H\left(  Y_{i}|W_{i}\right)  \label{KConTD1}\\
&  =%
{\textstyle\sum\limits_{i=1}^{n}}
\frac{1}{n}E_{i}\label{KConTD1a}\\
&  \leq%
{\textstyle\sum\limits_{i=1}^{n}}
\frac{1}{n}\Gamma\left(  D_{1i},D_{2i}\right)  \label{KConTD2}\\
&  \leq\Gamma_{I}\left(  D_{1},D_{2}\right)
\end{align}
where (\ref{KConTD1}) follows from the definition of $W_{1,i}$,
(\ref{KConTD1a}) and (\ref{KConTD2}) follow from (\ref{GammaRDE_I}) in
Definition \ref{DefRDEKaspi} and from Lemma \ref{Lemma_TD}.

\textit{Achievability}: Fix $p\left(  w_{1},w_{2}|x,y\right)  $. First
generate $M_{1}=2^{n\left(  I(W_{1};XY)+\epsilon\right)  }$, $W_{1}^{n}\left(
j_{1}\right)  $ sequences, $j_{1}=1,2,\ldots,M_{1}$, i.i.d. according to
$p\left(  w_{1}\right)  $ (obtained from $p\left(  w_{1},w_{2}|x,y\right)  )$.
Generate $M_{2}=2^{n\left(  I(W_{2};XYW_{1})+\epsilon\right)  }$ $W_{2}%
^{n}\left(  j_{2}\right)  $ sequences i.i.d. according to $p\left(
w_{2}\right)  $ (obtained from $p\left(  w_{1},w_{2}|x,y\right)  )$. Bin the
resulting $W_{2}^{n}$ sequences into $S$ bins (analogously to the Wyner-Ziv
binning), chosen at random where $S=2^{n\left(  I(XYW_{1};W_{2})-I(W_{1}%
Y;W_{2})+\epsilon\right)  }$, and index these bins as $b\left(  j_{2}\right)
$. Upon observing a source sequence $\left(  x^{n},y^{n}\right)  ,$ the
encoder searches for a $W_{1}^{n}\left(  j_{1}\right)  $ sequence such that
$\left(  x^{n},y^{n},w_{1}^{n}\left(  j_{1}\right)  \right)  \in
\mathcal{T}_{XYW_{1}}\left(  n,\epsilon\right)  $ (the choice of $M_{1}$
ensures that there exists at least one such $j_{1}$). Next, the encoder
searches for a $w_{2}^{n}\left(  j_{2}\right)  $ such that $\left(
x^{n},y^{n},w_{1}^{n}\left(  j_{1}\right)  ,w_{2}^{n}\left(  j_{2}\right)
\right)  \in\mathcal{T}_{XYW_{1}W_{2}}\left(  n,\epsilon\right)  $ (the choice
of $M_{2}$ ensures that there exists at least one such $j_{2}$). The encoder
sends $\left(  j_{1},b\left(  j_{2}\right)  \right)  $ where $b\left(
j_{2}\right)  $ is the bin index of the $w_{2}^{n}\left(  j_{2}\right)  $
sequence at a rate $R=I(XY;W_{1})+I(X;W_{2}|W_{1}Y)+\epsilon$. Thus, we have
\end{subequations}
\begin{align}
&  \Pr\left(  B=b\left(  j_{2}\right)  |\left(  x^{n},y^{n},w_{1}^{n}\left(
j_{1}\right)  ,w_{2}^{n}\left(  j_{2}\right)  \right)  \right.  \left.
\in\mathcal{T}_{XYW_{1}W_{2}}\left(  n,\epsilon\right)  \right)  \nonumber\\
&  =\Pr\left(  B=b\left(  j_{2}\right)  |w_{2}^{n}\left(  j_{2}\right)
\right.  \left.  \in\mathcal{T}_{W_{2}}\left(  n,\epsilon\right)  \right)
=1/S.\label{ProbBKas}%
\end{align}
where (\ref{ProbBKas}) is the result of the code construction which yields a
Markov chain relationship $\left(  XYW_{1}\right)  -W_{2}-B$. With $\mu$ as
defined in (\ref{mu_def}) for the typical set $\mathcal{T}_{XYW_{1}W_{2}}$,
and $J\equiv\left(  J_{1},B\left(  J_{2}\right)  \right)  $, the achievable
equivocation can be lower bounded as
\begin{subequations}
\label{KasHY}%
\begin{align}
&  \frac{1}{n}H\left(  Y^{n}|J_{1},B\left(  J_{2}\right)  \right)  \nonumber\\
&  \geq\frac{1}{n}H\left(  Y^{n}|J_{1},B\left(  J_{2}\right)  ,\mu\right)  \\
&  =\frac{1}{n}H\left(  Y^{n}|W_{1}^{n}\left(  J_{1}\right)  ,B\left(
J_{2}\right)  ,\mu\right)  \\
&  \geq\Pr\left(  \mu=1\right)  \frac{1}{n}H\left(  Y^{n}|W_{1}^{n}\left(
J_{1}\right)  ,B\left(  J_{2}\right)  ,\mu=1\right)  .\label{KasHYach}%
\end{align}
The probability $\Pr\left(  y^{n}|w_{1}^{n}\left(  j_{1}\right)  ,b\left(
j_{2}\right)  ,\mu=1\right)  $ for all $j_{1},j_{2},$ and $y^{n}$ can be
written as
\end{subequations}
\begin{subequations}
\begin{align}
&
{\textstyle\sum\limits_{w_{2}^{n}}}
\Pr\left(  y^{n},w_{2}^{n}|w_{1}^{n}\left(  j_{1}\right)  ,b\left(
j_{2}\right)  ,\mu=1\right)  \nonumber\\
&  =%
{\textstyle\sum\limits_{w_{2}^{n}}}
\Pr\left(  w_{2}^{n}|w_{1}^{n}\left(  j_{1}\right)  ,b\left(  j_{2}\right)
,\mu=1\right)  \Pr\left(  y^{n}|w_{1}^{n}\left(  j_{1}\right)  ,w_{2}^{n}%
,\mu=1\right)  \label{KasPrY2}%
\end{align}
where (\ref{KasPrY2}) follows from the fact that $\left(  XYW_{1}\right)
-W_{2}-B$ forms a\ Markov chain. The probability $\Pr\left(  w_{2}%
|w_{1}\left(  j_{1}\right)  ,b\left(  j_{2}\right)  ,\mu=1\right)  $ can be
rewritten as
\end{subequations}
\begin{align}
&  \frac{\Pr\left(  w_{2}^{n},w_{1}^{n}\left(  j_{1}\right)  ,b\left(
j_{2}\right)  |\mu=1\right)  }{\Pr\left(  w_{1}^{n}\left(  j_{1}\right)
,b\left(  j_{2}\right)  |\mu=1\right)  }\nonumber\\
&  =\frac{\Pr\left(  w_{2}^{n},w_{1}^{n}\left(  j_{1}\right)  |\mu=1\right)
/\left\vert S\right\vert }{%
{\textstyle\sum\limits_{w_{2}}}
\Pr\left(  w_{2}^{n},w_{1}^{n}\left(  j_{1}\right)  |\mu=1\right)  /\left\vert
S\right\vert }\label{Prw1w2}\\
&  =\Pr\left(  w_{2}^{n}|w_{1}^{n}\left(  j_{1}\right)  ,\mu=1\right)
.\label{Prw1w2_2}%
\end{align}
Substituting (\ref{Prw1w2_2}) in (\ref{KasPrY2}), $\Pr\left(  y^{n}|w_{1}%
^{n}\left(  j_{1}\right)  ,b\left(  j_{2}\right)  ,\mu=1\right)  $ can be
written as
\begin{subequations}
\label{KasPr_w2w1}%
\begin{align}
&
{\textstyle\sum\limits_{w_{2}^{n}}}
\Pr\left(  w_{2}^{n}|w_{1}^{n}\left(  j_{1}\right)  ,\mu=1\right)  \Pr\left(
y^{n}|w_{1}^{n}\left(  j_{1}\right)  ,w_{2}^{n},\mu=1\right)  \nonumber\\
&  =%
{\textstyle\sum\limits_{w_{2}^{n}}}
\Pr\left(  y^{n},w_{2}^{n}|w_{1}^{n}\left(  j_{1}\right)  ,\mu=1\right)  \\
&  =\Pr\left(  y^{n}|w_{1}^{n}\left(  j_{1}\right)  ,\mu=1\right)  \\
&  \leq2^{-n\left(  H\left(  Y|W_{1}\right)  -\epsilon\right)  }%
\label{KasPrYW1}%
\end{align}
where we have used the fact that for a typical set $\Pr\left(  \mathcal{T}%
_{YW_{1}W_{2}}\left(  n,\epsilon\right)  \right)  \geq\left(  1-\epsilon
\right)  $ \cite[chap. 2]{Cover:book}. From (\ref{KasHYach}) and
(\ref{KasPrYW1}), we then have
\end{subequations}
\begin{subequations}
\label{KasH}%
\begin{align}
H\left(  Y^{n}|w_{1}^{n}\left(  j_{1}\right)  ,b\left(  j_{2}\right)
,\mu=1\right)   &  =%
{\textstyle\sum\limits_{y^{n}}}
\Pr\left(  y^{n}|w_{1}^{n}\left(  j_{1}\right)  ,\mu=1\right)  \log\frac
{1}{\Pr\left(  y^{n}|w_{1}^{n}\left(  j_{1}\right)  ,\mu=1\right)  }\\
&  \geq%
{\textstyle\sum\limits_{y^{n}}}
\Pr\left(  y^{n}|w_{1}^{n}\left(  j_{1}\right)  ,\mu=1\right)  n\left(
H\left(  Y|W_{1}\right)  -\epsilon\right)  \label{KasHYJexp}\\
&  \geq n\left(  1-\epsilon\right)  \left(  H\left(  Y|W_{1}\right)
-\epsilon\right)
\end{align}
where in (\ref{KasHYJexp}) we have used the fact that for a typical set
$\Pr\left(  T_{YW_{1}W_{2}}\left(  n,\epsilon\right)  \right)  \geq\left(
1-\epsilon\right)  $ \cite[chap. 2]{Cover:book}. Thus, we have
\end{subequations}
\begin{align}
\frac{1}{n}H\left(  Y^{n}|J\right)   &  \geq\Pr\left(  \mu=1\right)  \frac
{1}{n}%
{\textstyle\sum\limits_{w_{1}^{n},b\left(  j_{2}\right)  }}
\Pr\left(  w_{1}^{n}\left(  j_{1}\right)  ,b\left(  j_{2}\right)
|\mu=1\right)  H\left(  Y^{n}|w_{1}^{n}\left(  j_{1}\right)  ,b\left(
j_{2}\right)  ,\mu=1\right)  \\
&  \geq\left(  1-\epsilon\right)  ^{3}\left(  H\left(  Y|W_{1}\right)
-\epsilon\right)
\end{align}
where we have used the fact that for a typical set $\Pr\left(  \mathcal{T}%
_{YW_{1}W_{2}}\left(  n,\epsilon\right)  \right)  \geq\left(  1-\epsilon
\right)  $ \cite[chap. 2]{Cover:book}. The proof concludes by observing that
$H\left(  Y^{n}\right)  \geq H\left(  Y^{n}|J\right)  $ and $\epsilon
\rightarrow0$ as $n\rightarrow\infty$.

\subsection{Converse Proof for region $\mathcal{G}_{4}$}

We start by a simple lower bound on the rate
\begin{align}
R &  \geq I(X,Y;W_{1})+I(X;W_{2}|W_{1},Y)\nonumber\\
&  \geq I(X,Y;\hat{X}_{1})\label{rateLB}%
\end{align}
and an upper bound on $\Gamma$
\begin{align}
\Gamma &  \leq H(Y|W_{1})\nonumber\\
&  =H(Y|W_{1},\hat{X}_{1})\nonumber\\
&  \leq H(Y|\hat{X}_{1})\nonumber\\
&  =H(Y)-I(Y;\hat{X}_{1}).\label{eqLB}%
\end{align}
We will now use the distortion constraint of decoder $1$ alone to
simultaneously lower bound the rate and upper bound the equivocation. Consider
an arbitrary $p^{(1)}(\hat{x}_{1}|x,y)$ (and denote this as distribution
$\mathcal{P}_{1}$) given as:
\begin{align}
p^{(1)}(0|0,0) &  =a,\quad p^{(1)}(0|1,1)=b\nonumber\\
p^{(1)}(0|0,E) &  =c,\quad p^{(1)}(0|1,E)=d.\nonumber
\end{align}
For this distribution, we have
\begin{align}
\mathbb{P}(X\neq\hat{X}_{1}) &  =(1/2)\Big[(1-p)(1-a+b)+p(1-c+d)\Big]\\
H(\hat{X}_{1}) &  =h\left(  \frac{1}{2}\big[(1-p)(a+b)+p(c+d)\big]\right)  \\
H(\hat{X}_{1}|X,Y) &  =\frac{(1-p)}{2}(h(a)+h(b))+\frac{p}{2}(h(c)+h(d))\\
H(\hat{X}_{1}|Y) &  =\frac{(1-p)}{2}(h(a)+h(b))+ph\left(  \frac{c+d}%
{2}\right)  .
\end{align}
These four quantities characterize the bounds in (\ref{rateLB}) and
(\ref{eqLB}) exactly and also the achievable distortion.

Now consider a new distribution $\mathcal{P}_{2}$, with conditional
probabilities as follows:
\begin{align}
p^{(2)}(0|0,0) &  =1-b,\quad p^{(2)}(0|1,1)=1-a\nonumber\\
p^{(2)}(0|0,E) &  =1-d,\quad p^{(2)}(0|1,E)=1-c.\nonumber
\end{align}
It is straightforward to verify that the distortion, rate and equivocation
terms are the same for both $\mathcal{P}_{1}$ and $\mathcal{P}_{2}$. Next,
define a new distribution $\mathcal{P}_{3}$ as follows:
\[
p^{(3)}(\hat{x}_{1}|x,y)=%
\begin{cases}
p^{(1)}(\hat{x}_{1}|x,y) & \mbox{ w.p. }1/2,\\
p^{(2)}(\hat{x}_{1}|x,y) & \mbox{ w.p. }1/2.
\end{cases}
\]

We now note that $I(X,Y;\hat{X}_{1})$ is convex in $p(\hat{x}_{1}|x,y)$ and
$H(Y|\hat{X}_{1})=H(Y)-I(Y;\hat{X}_{1})$ is concave in $p(\hat{x}_{1}|y)$. By
Jensen's inequality, this implies that the distribution $\mathcal{P}_{3}$
defined above uses a rate that is at most as large and leads to an
equivocation that is at least as large when compared to both the distributions
$\mathcal{P}_{1}$ and $\mathcal{P}_{2}$. Hence, it suffices to consider input
distributions of the form $p^{(3)}(\hat{x}_{1}|x,y)$, which can be explicitly
written as
\begin{align}
p^{(3)}(0|0,0) &  =1-\beta,\quad p^{(3)}(0|1,1)=\beta\nonumber\\
p^{(3)}(0|0,E) &  =1-\alpha,\quad p^{(3)}(0|1,E)=\alpha.\nonumber
\end{align}
To satisfy the distortion constraint, we also have
\[
D_{1}\geq(1-p)\beta+p\alpha
\]
which leads to $\beta=(D_{1}-p\alpha)/(1-p)$. Now, also note that for a fixed
$\alpha$, this scheme yields a distortion of $p\alpha$ at the decoder $2$.
Furthermore, since the range of $\alpha\in\lbrack0,D_{1}/p]$, we note that the
worst case distortion for decoder $2$ (for a fixed $D_{1}$) is $pD_{1}%
/p=D_{1}$. This implies that as long as
\[
D_{2}\geq D_{1}%
\]
this region yields the stated tradeoff for the region $\mathcal{G}_{4}$.

\bibliographystyle{IEEEtran}
\bibliography{refHB}

\end{document}